\documentclass[lettersize,journal]{IEEEtran}
\usepackage{amsmath,amssymb,amsfonts}
\usepackage{arydshln}
\usepackage{graphicx}
\usepackage{textcomp}
\usepackage{wrapfig}

\usepackage{algorithm}
\usepackage{algpseudocode}
\usepackage{color}
\usepackage{xcolor}
\usepackage{array}
\usepackage[caption=false,font=footnotesize,labelfont=rm,textfont=sf]{subfig}
\usepackage[justification=centering]{caption}

\usepackage{stfloats}
\usepackage{url}
\usepackage{verbatim}
\usepackage{graphicx}
\usepackage{subfig}
\usepackage{epstopdf}
\usepackage{balance}
\usepackage{booktabs}

\usepackage{stmaryrd}

\newtheorem{definition}{\bf Definition}

\newtheorem{theorem}{\bf Theorem}

\usepackage[numbers,sort&compress]{natbib}

\def\BibTeX{{\rm B\kern-.05em{\sc i\kern-.025em b}\kern-.08em
    T\kern-.1667em\lower.7ex\hbox{E}\kern-.125emX}}
\usepackage{balance}

\begin{document}
\title{Spatial Angular Pseudo-Derivative Search Algorithm: A Real-Time Single-Snapshot Super-Resolution Sparse DOA Scheme for Automotive Radar}
\author{Longxin Bai, Jingchao Zhang, and Liyan Qiao
\thanks{
	This work has been submitted to the IEEE for possible publication. Copyright may be transferred without notice, after which this version may no longer be accessible.
	
	This work is supported by National Natural Science Foundation of China
	under Grant 61701138, the Natural Science Foundation of Heilong jiang
	Province of China under Grant LH2022F019, Hei Long Jiang Postdoctoral
	Foundation under Grant LBH-Z16087 and the Fundamental Research Funds
	for the Central Universities under Grant HIT.NSRIF202339. (Corresponding author: Jingchao Zhang)
	
L. Bai, J. Zhang, and L. Qiao are with School of Electronics and Information Engineering, Harbin Institute of Technology, Harbin 150001, China.(e-mail: bailongxin@stu.hit.edu.cn, zhangjingchao@hit.edu.cn, qiaoliyan@hit.edu.cn)} }


\maketitle

\begin{abstract}
Accurate, high-resolution, and real-time DOA estimation plays a crucial role in automotive radar perception. While sparse signal recovery techniques offer super-resolution and high-precision estimation, their prohibitive computational complexity remains a primary bottleneck for practical deployment. This paper proposes a sparse DOA estimation scheme specifically tailored for the stringent requirements of automotive radar such as limited computational resources, restricted array apertures, and single-snapshot constraints. By leveraging the spatial angular pseudo-derivative (SAPD) property of the sparse DOA solutions and incorporating this property as a constraint into an $\ell_0$-norm minimization problem, the proposed formulation transforms unordered exhaustive verification of candidate solutions into an ordered search process for sparse DOA estimation. Thus, the associated solver, called the SAPD Search algorithm, naturally transforms the high-dimensional optimization task into an efficient grid-search scheme. The SAPD search algorithm circumvents high-order matrix inversions and computationally intensive iterations. We also provide an analysis of the computational complexity of the proposed algorithm. Numerical simulations and experimental validation demonstrate that the SAPD Search algorithm achieves a superior balance of millisecond-level computational efficiency, high precision, and super-resolution, making it highly suitable for next-generation automotive radar applications.
\end{abstract}


\section{INTRODUCTION}
While super-resolution direction-of-arrival (DOA) estimation has been extensively studied for over four decades \cite{96-Twodecades}, \cite{23-Twodecades}, simultaneously achieving real-time operation and high estimation accuracy remains a formidable challenge under stringent constraints such as limited computational resources, restricted array apertures, and the single-snapshot constraint. Notably, this technology is pivotal for automotive millimeter-wave (mmWave) radars, where it enables critical perception tasks within autonomous driving systems and robotic environments \cite{JSTSP-4D-Sparsity}, \cite{TITS-4D-Cloud-Points}, \cite{Radar_Localization}, \cite{RadarSLAM}, \cite{Radar_Point_Clouds_Generation}. Therefore, addressing these challenges to enhance the angular estimation performance of mmWave radars is essential to meet the requirements \cite{Radar_Automotive} of next-generation autonomous driving.

In autonomous driving scenarios, DOA estimation must be performed for all Range–Doppler (RD) cells identified by Constant False Alarm Rate (CFAR) detection within the time interval between two consecutive radar frames \cite{Blind_Array_Calibration}, \cite{Scater_to_Point_Clouds}, \cite{GPU_Real_Time_DOA}. Failure to complete the estimation within the available time budget may lead to the loss of target information, degrade point-cloud quality, and consequently affect the reliability of subsequent radar processing. Considering the computational overhead of downstream radar processing following point-cloud generation, computational efficiency becomes a primary prerequisite for the practical deployment of super-resolution DOA algorithms. Meanwhile, the estimation accuracy and angular resolution of DOA algorithms are also key factors affecting point-cloud quality and the performance of subsequent radar processing. Therefore, a practical super-resolution DOA method should simultaneously achieve low computational complexity, high estimation accuracy, and superior angular resolution \cite{Radar_Automotive}.

Conventional beamforming methods can achieve real-time DOA estimation from a single snapshot, but their angular resolution is inherently limited. In contrast, subspace-based methods, such as MUSIC \cite{MUSIC} and ESPRIT \cite{ESPRIT}, provide super-resolution capability but typically rely on multiple snapshots to estimate the covariance matrix. As a result, their applicability is limited in highly dynamic radar scenarios \cite{Radar_Automotive}.

In contrast, sparse DOA estimation methods can achieve super-resolution estimation under the single-snapshot scenario and exhibit good robustness to correlated signals. As a result, they are widely regarded as one of the most promising DOA estimation frameworks for mobile-platform applications. However, their superior estimation performance is typically achieved at the cost of increased computational complexity \cite{23-Twodecades}.

When the estimated DOAs correspond to discrete points on a predefined spatial grid, such methods are generally referred to as on-grid methods. Starting from the formulation of the $\ell_0$-norm minimization problem \cite{Sparse_Approximate_exs}, numerous approaches have been developed based on convex relaxation \cite{l1svd} and nonconvex penalty functions \cite{CEL0_l1}, \cite{amir2021trimmed}. Although these methods can achieve high estimation accuracy and angular resolution, their solution procedures typically require repeated iterative optimization and matrix operations involving an overcomplete dictionary. Since the dimension of the dictionary is usually much larger than that of the array manifold, the resulting computational complexity remains high, making real-time implementation challenging. To reduce the computational burden, various studies have proposed more efficient optimization strategies \cite{FISTA}, \cite{Fast_LS_ADMM} or introduced weighted $\ell_1$ formulations \cite{IRL1_Wip} to accelerate convergence. However, such approaches generally fail to significantly reduce the overall computational cost associated with the overcomplete dictionary.

Furthermore, Orthogonal Matching Pursuit (OMP) \cite{OMP_Current} employs a greedy strategy, thereby avoiding the iterative optimization procedures required by convex and nonconvex formulations and significantly reducing computational complexity. However, in our simulations, its estimation accuracy and angular resolution are insufficient for the high-resolution DOA estimation task considered in this work.

Another representative class of methods is based on covariance fitting, such as SPICE \cite{SPICE} and IAA-APES \cite{IAA_APES}. SPICE \cite{SPICE}, \cite{MMSE_Framework} performs DOA estimation through covariance fitting and achieves high computational efficiency. However, its performance is highly dependent on the number of snapshots, and its angular resolution and estimation accuracy degrade significantly under the single-snapshot scenario. In contrast, IAA-APES iteratively fits the covariance matrix using a weighted least-squares criterion, enabling good stability, estimation accuracy, and angular resolution under single-snapshot conditions. Since it also avoids the iterative computations associated with overcomplete dictionaries, its computational complexity is generally lower than that of most convex and nonconvex optimization-based approaches. Nevertheless, IAA-APES still requires multiple iterations to complete the covariance fitting process, and each iteration involves a covariance matrix inversion. As a result, its computational cost remains considerable when processing a large number of RD cells in automotive radar applications.

In practice, the true DOAs rarely coincide exactly with the predefined spatial grid, giving rise to grid mismatch errors. In general compressed sensing problems, dictionary mismatch can be characterized using perturbation models \cite{Orig_Off_grid}, \cite{Joint_Sparse_Recovery}. For DOA estimation, many existing studies model the resulting off-grid bias using a first-order Taylor expansion \cite{OGSBI}.

Within the deterministic optimization framework, the off-grid bias is typically estimated jointly with the sparse coefficients, which can be achieved through approaches such as gradient descent \cite{Off-grid_Chen} or block-sparse recovery \cite{block_sparse}. Although these methods can alleviate the performance degradation caused by grid mismatch, they still rely on sparse optimization frameworks constructed from overcomplete dictionaries and require the additional estimation of off-grid bias, resulting in relatively high computational complexity.

Within the Sparse Bayesian Learning (SBL) framework, the off-grid bias can be jointly estimated with the sparse coefficients through Bayesian inference \cite{Cluster_OGSBI}, \cite{SBL_HCP}, \cite{OGSBI}, \cite{Hierarchical_SBI}. Although such methods generally provide excellent estimation performance, Bayesian inference itself is computationally intensive. To reduce computational complexity, various acceleration strategies have been proposed. However, these methods still fundamentally rely on sparse representation frameworks constructed from overcomplete dictionaries \cite{GE}, . Inspired by the off-grid modeling concept, Grid Evolution (GE) \cite{GE} and related methods further reduce part of the computational burden by dynamically refining the grid. Nevertheless, their underlying implementations still depend on the Root-OGSBI \cite{Root-OGSBI} framework, and thus the resulting complexity reduction remains limited.

Gridless DOA estimation methods \cite{Atom_min_orig}, \cite{Gridless_Current}, on the other hand, can theoretically eliminate the grid mismatch problem caused by discretization. However, their solution procedures often rely on semidefinite programming (SDP), resulting in even higher computational complexity. Although considerable efforts have been devoted to reducing the computational burden through improved SDP solvers, the resulting complexity reduction remains limited. To address this issue, methods such as VALSE \cite{VALSE} avoid SDP optimization by employing variational inference and heuristic search strategies, thereby significantly reducing computational complexity while maintaining favorable estimation performance. Nevertheless, for automotive radar applications involving real-time processing of a large number of RD cells, the computational cost remains relatively high.

In addition, methods based on GPU acceleration and deep learning can also achieve real-time single-snapshot DOA estimation \cite{GPU_Real_Time_DOA}. However, their performance improvements typically rely on additional computational resources. This work focuses on high-resolution DOA estimation under limited computational resources. Therefore, such hardware-dependent approaches are beyond the scope of this paper.

From the above analysis, it can be observed that the major computational burden of existing high-performance single-snapshot DOA estimation methods mainly arises from large-scale matrix operations. Therefore, reducing the scale of matrix computations during the iterative process and minimizing the use of high-dimensional matrix inversions are important directions for lowering computational complexity. However, most existing sparse DOA estimation methods are formulated from the perspective of $\ell_0$-norm minimization. Although this framework provides a general formulation for sparse recovery, it is essentially a generic sparse reconstruction model and does not fully exploit the physical structure inherent to the DOA estimation problem. Consequently, further reducing the computational complexity within this framework alone is often difficult. To achieve lower computational complexity while preserving super-resolution capability and facilitating subsequent refinement of estimation accuracy, it is necessary to introduce additional structural constraints into the original sparse recovery model.

To avoid the high computational complexity introduced by the involvement of large-scale matrices, such as full overcomplete dictionary and covariance-like matrices, in sparse DOA estimation, thereby achieving low-complexity, high-accuracy, and super-resolution DOA estimation while bridging the gap between sparse signal recovery theory and practical applications, this paper proposes a sparse DOA estimation method for automotive millimeter-wave radar systems.

\textbf{1) Spatial Angular Pseudo-Derivative (SAPD) Property:} By analyzing the relationship between the sparse DOA solution and the spatial discrete grid, the SAPD property of the sparse DOA solution is characterized, and the corresponding SAPD constraints are further proposed. This property reveals the intrinsic relationship between the local minima of the corresponding subproblem and the spatial discrete grid under a fixed sparsity level, providing the design basis for the subsequent Sparse DOA formulation and corresponding efficient search strategy.

\textbf{2) SAPD-Constrained Sparse DOA Optimization Function:} The SAPD Constraints are introduced into the conventional $\ell_0$-norm sparse recovery model as structural constraints, leading to a new Sparse DOA optimization function. As a result, the spatial structural information characterized by the SAPD Property is incorporated into the proposed Sparse DOA formulation. Although the proposed optimization function still adopts the $\ell_0$-norm minimization objective, it transforms the original combinatorial optimization problem, which requires unordered exhaustive verification of candidate solution combinations, into an ordered search problem for sparse DOA estimation. Consequently, the proposed formulation leads to an efficient search strategy that avoids large-scale matrix computations.

\textbf{3) Efficient SAPD Search Algorithm:} Based on the proposed Sparse DOA formulation and the SAPD Property of the sparse DOA solution, an efficient SAPD Search Algorithm is developed. By exploiting the information provided by the beamforming spatial spectrum together with the SAPD Property, the proposed algorithm solves the established Sparse DOA formulation through a grid search strategy. Consequently, it avoids the use of an overcomplete dictionary, large-scale matrix inversions, and computationally intensive iterative optimization, thereby significantly reducing the computational complexity of DOA estimation. Meanwhile, the proposed method does not require prior knowledge of the number of incident sources.

The remainder of this paper is organized as follows. Section \ref{Sec2} introduces the sparse DOA estimation problem and derives the expression of the off-grid bias under a given angular grid. Section \ref{Sec3} presents the concept of the Spatial Angular Pseudo-Derivative (SAPD) and formulates a new sparse DOA optimization problem based on SAPD. Section \ref{Sec4} describes the proposed SAPD Search Algorithm in detail. Section \ref{Sec6} presents numerical simulations and experimental results. Finally, Section \ref{Sec7} concludes the paper.

Notations: $\| \cdot \|_2$, $\| \cdot \|_1$ denote the $\ell_2$-norm and $\ell_1$-norm, respectively. $(\cdot)^T$, $(\cdot)^H$, $(\cdot)^{-1}$ and $(\cdot)^{\dagger}$ represent the transposition, Hermitian transposition, inversion, pseudo-inversion, respectively. $\lceil \cdot \rceil$ and $\operatorname{round}(\cdot)$ denote the ceiling and rounding operators, respectively. $\mathbb{E}(\cdot)$ denotes the expectation operator. $j = \sqrt{1}$. $\mathbf{1}_N$ is the  $N \times 1$ all-ones vector. For any $n \in \mathbb{Z}_{+}$, we denote $[n] \triangleq \{ 1, \dots n \}$.  $\operatorname{sgn}(\cdot)$ denotes the signum function. The notation $\vert \cdot \vert$ denotes the absolute value when applied to a scalar, the element-wise absolute value when applied to a vector, and the cardinality when applied to a set. The operator $\operatorname{sgn}(\cdot)$ denotes the signum function. When applied to a vector, it operates element-wise. For an index set $\mathcal{G}=\{ g_1,g_2,\dots,g_d \}$, the subvector of $\boldsymbol{x}$ associated with $\mathcal{G}$ is defined as
\begin{align}
	\boldsymbol{x}_\mathcal{G} &= [x_{g_1}, x_{g_2}, \dots, x_{g_d}]^T, \ g_i \in \mathcal{G}.
\end{align}

\section{Sparse DOA Model and Off-Grid Bias Formulation} \label{Sec2}

Consider a frequency-modulated continuous-wave (FMCW) radar equipped with a uniform linear array (ULA) consisting of $M$ receive antennas with inter-element spacing $d_a =\lambda/2$, where $\lambda$ denotes the carrier wavelength. The array receives echoes from $K$ far-field sources impinging from directions $\boldsymbol{\vartheta}^{*}=[\theta_1^{*},\theta_2^{*},\cdots,\theta_K^{*}]$, where the superscript ${\theta}_k^{*}$ denotes the $k$-th true direction-of-arrival (DOA).

Under the narrowband far-field assumption, the steering vector corresponding to the $k$th source is given by
\begin{equation}
	\boldsymbol{a}(\theta_k^*) = [1, e^{-j\pi\sin{(\theta_k^*)}}, \cdots, e^{-j\pi(M-1)  \sin{(\theta_k^*)}}]^T,  \label{steering_vector}
\end{equation}
where $\theta_k^* \in \boldsymbol{\vartheta}^*$, for $k = 1, \dots, K$. Assuming an ideal array, the received data vector can be expressed as
\begin{equation}
	\boldsymbol{y} = \sum_{k=1}^{K} {\boldsymbol{a}}(\theta_{k}^*) s_k +{\boldsymbol{n}} 
	= \mathbf{A}(\boldsymbol{\vartheta}^*)\boldsymbol{s} + \boldsymbol{n},
\end{equation}
where $\mathbf{A}(\boldsymbol{\vartheta}^*) = [\boldsymbol{a}(\theta_1^*), \boldsymbol{a}(\theta_2^*), \cdots, \boldsymbol{a}(\theta_K^*)] \in \mathbb{C}^{M \times K}$  is the array manifold matrix, $\boldsymbol{s} = [s_1, s_2, \dots, s_K]^T \in \mathbb{C}^{K}$ denotes the signal vector, and $\boldsymbol{n} \in \mathbb{C}^{M}$ represents the additive white Gaussian noise vector. 

Subsequently, a uniform angular grid is constructed as $\boldsymbol{\vartheta} = [\theta_1, \theta_2, \dots, \theta_G]$ with grid interval $\Delta\theta$. Based on this discretized grid, the overcomplete array manifold matrix is defined as
$\mathbf{A}(\boldsymbol{\vartheta}) = [\boldsymbol{a}(\theta_1), \boldsymbol{a}(\theta_2), \dots, \boldsymbol{a}(\theta_G)] \in \mathbb{C}^{M \times G}$. Accordingly, the DOA observation model can be rewritten as $\boldsymbol{y} = \mathbf{A}(\boldsymbol{\vartheta})\boldsymbol{x} + \boldsymbol{n}$, where $\boldsymbol{x} \in \mathbb{C}^{G}$ denotes a sparse representation vector. Considering the presence of noise, the sparse signal recovery problem can be formulated as the following $\ell_0$-norm minimization problem
\begin{equation}
	\hat{\boldsymbol{x}} = \mathop{\arg\min}\limits_{\boldsymbol{x}} \| \boldsymbol{x} \|_0  \quad  \text{s.t.} \quad \| \boldsymbol{y} - \mathbf{A}(\boldsymbol{\vartheta}) \boldsymbol{x} \|_2^2 \le \epsilon, \label{l0_obj}
\end{equation}
where $\epsilon_e$ accounts for recovery tolerance for exhaustive search. Problem \eqref{l0_obj} is a combinatorial optimization problem which can be solved via exhaustive search \cite{Sparse_Approximate_exs}, \cite{Mathematical_Compressive}.

Although the overall computational cost of exhaustive search is prohibitive, its solving procedure reveals an important characteristic. For a given index set, the verification process only requires obtaining the corresponding recovered coefficients through least-squares (LS) estimation and comparing the recovery error with the predefined recovery tolerance to determine whether the candidate solution satisfies the constraint of \eqref{l0_obj}.

For the DOA estimation problem, since a ULA with $M$ sensors can resolve at most $M-1$ sources, let $\Lambda_s = \{ \mathcal{G}_i | \mathcal{G}_i \subset [G], \ 1 \leq \vert \mathcal{G}_i \vert \leq M - 1 \}$ denote the set of all candidate support sets. The elements of $\mathcal{G}_i$ are denoted by $\mathcal{G}_i = \{ g_1, g_2, \dots, g_d  \} \subset [G]$, where $d$ is the cardinate of $\mathcal{G}_i$. For a given candidate support set $\mathcal{G}_i$, let $\mathcal{G}_i^c = [G] - \mathcal{G}_i$ and  $\boldsymbol{x}_{\mathcal{G}_i^c} = \mathbf{0}$. The corresponding coefficient vector $\boldsymbol{x}_{\mathcal{G}_i}$ is obtained by the least-squares estimation \cite{Mathematical_Compressive}, given by
\begin{align}
	\boldsymbol{x}_{\mathcal{G}_i} = \boldsymbol{x}(\boldsymbol{\vartheta}_{\mathcal{G}_i}) = (\mathbf{A}(\boldsymbol{\vartheta}_{\mathcal{G}_i})^H\mathbf{A}(\boldsymbol{\vartheta}_{\mathcal{G}_i}))^{-1} \mathbf{A}(\boldsymbol{\vartheta}_{\mathcal{G}_i})^H \boldsymbol{y}.   \label{x_theta_function}
\end{align}
Equation \eqref{x_theta_function} expresses $\boldsymbol{x}_{\mathcal{G}_i}$ as a function of the corresponding angle set $\boldsymbol{\vartheta}_{\mathcal{G}_i}$. $\| \boldsymbol{y} - \mathbf{A}(\boldsymbol{\vartheta}_{\mathcal{G}_i}) \boldsymbol{x}_{\mathcal{G}_i} \|_2^2 \le \epsilon$ the candidate support set is considered to satisfy the recovery error constraint and is retained for the subsequent final solution selection.

Through the above procedure, it can be observed that the matrix involved in Equation \eqref{x_theta_function} is only $\mathbf{A}(\boldsymbol{\vartheta}_{\mathcal{G}_i}) \in \mathbb{C}^{M \times d}$, where the computational dimension is determined by the size of the candidate support set $d$ rather than the total number of grid points $G$, with $d\ll G$. The dominant computational cost of Equation \eqref{x_theta_function} is approximately $O(Md^2)$. Therefore, the verification process in exhaustive search avoids direct operations involving the full overcomplete dictionary and covariance-like matrix.

However, although the above verification process has a relatively low computational cost, exhaustive search still cannot satisfy practical requirements. This is because the candidate support set $\Lambda_s$ is usually unordered, and the algorithm needs to perform least-squares estimation and recovery error verification for a large number of possible candidates individually, resulting in a significant increase in the overall computational burden.

Therefore, to exploit the computational advantage provided by the low-dimensional verification process, a new characteristic of the sparse DOA solution is required, which enables a structured search strategy and efficient verification. In this paper, this characteristic is represented by the mapping relationship between the solution of the DOA estimation problem and the spatial discrete grid points, and is defined as the solution Spatial Angular Pseudo-Derivative (SAPD) Property. Based on this property, a new objective function is further formulated, and the corresponding solving method is developed, which will be described in detail in the following sections.

\section{The concept of SAPD and SAPD-Constrained Sparse DOA Formulation} \label{Sec3}

This section first provides a detailed discussion of the Spatial Angular Pseudo-Derivative (SAPD) Property of sparse DOA solutions. Subsequently, the SAPD Constraints derived from this property are introduced and incorporated into the sparse DOA estimation objective function, leading to the proposed SAPD-Constrained Sparse DOA Formulation.

\subsection{The Concept of SAPD Property}

We first consider the case of $|\mathcal{G}|=K$ with $\mathcal{G}\subset [G]$, i.e., the cardinality of the given index set is equal to the number of incident sources. Without loss of generality, the true incident angles are assumed to satisfy $\theta_1^* < \theta_2^* < \cdots < \theta_K^*$. Accordingly, the indices of the grid points in $\boldsymbol{\vartheta}_{\mathcal{G}}$ are numbered in ascending order as $g_1< g_2 < \cdots < g_K$. For a given index set $\mathcal{G}$, the mapping relationship between $\mathcal{G}$ and the set of true incident angles $\boldsymbol{\theta}^{*}$ is defined as
\begin{equation}
	g_k = \mathop{\arg\min}\limits_{g_i \in \mathcal{G} } \vert \theta_{k}^* - \theta_{g_i} \vert.
\end{equation}
The above mapping indicates that, for each true incident angle $\theta_k^*$, the corresponding grid point $\theta_{g_k}$ within the given index set $\mathcal{G}$ is the nearest grid point to $\theta_k^*$. Based on this mapping relationship, the $k$-th true incident angle can be expressed as
\begin{equation}
	\theta_k^* = \theta_{g_k} + (\theta_{k}^* -  \theta_{g_k}) = \theta_{g_k} + \beta_k, \ \theta_{g_k} \in \boldsymbol{\vartheta}_{\mathcal{G}} \subset \boldsymbol{\vartheta}. \label{bias_model}
\end{equation}  
where $\beta_k=\theta_k^*-\theta_{g_k}$ denotes the bias parameter of the true incident angle $\theta_k^*$ with respect to its corresponding grid point $\theta_{g_k}$.

\textbf{Remark:} It should be noted that the bias parameter defined in this paper is fundamentally different from the conventional off-grid parameter. The conventional off-grid parameter is usually defined based on the complete spatial discretized grid $\boldsymbol{\vartheta}$, where the grid point closest to the true incident angle is first identified over the entire grid, and the offset between the true angle and the corresponding grid point is then regarded as the interpolation parameter. However, for a given index set $\mathcal{G}\subset[G]$, the associated grid point is not necessarily the nearest grid point to the true incident angle within the complete discretized grid $\boldsymbol{\vartheta}$. Therefore, in this paper, the bias parameter is introduced to characterize the angular offset between the true incident angle and the corresponding grid point selected from the current index set $\mathcal{G}$, instead of adopting the conventional definition of the off-grid parameter. Therefore, according to the definition of the bias parameter, the conventional off-grid parameter can be regarded as a special case of the bias parameter under a specific grid selection condition.

According to the definition of the bias parameter in \eqref{bias_model}, for a target angle $\theta_k^*$, the magnitude of the bias $|\beta_k|$ is determined by the angular distance between the grid point and $\theta_k^*$. Furthermore,
\begin{align}
	\theta_{g_k} < \theta_k^* \ \Rightarrow \ \beta_{k} > 0 , \quad \theta_{g_k} > \theta_k^* \ \Rightarrow \  \beta_{k} < 0. \label{SAPD_PL}
\end{align}
Therefore, both the magnitude and the sign of the bias vary around the target angle. Specifically, the magnitude of the bias reaches its minimum near the true incident angle, changes its sign when crossing the target angle, and increases with the angular distance from the true incident angle. This behavior exhibits a first-order derivative-like characteristic around the stationary point of a quadratic function. Motivated by this derivative-like characteristic in the spatial angular domain, we refer to this property of the sparse DOA solution as the \textbf{Spatial Angular Pseudo-Derivative (SAPD)} property. 

For convenience, the bias parameter vector associated with a given angle set $\boldsymbol{\vartheta}_{\mathcal{G}}$ is denoted as $\boldsymbol{\beta}(\boldsymbol{\vartheta}_{\mathcal{G}}) = \boldsymbol{\beta}_{\mathcal{G}} = [\beta_1, \dots, \beta_K]^T$. Since the true incident angles $\boldsymbol{\vartheta}^{*}$ do not necessarily lie on the discretized spatial grid $\boldsymbol{\vartheta}$, let $\mathcal{G}^{*}$ denote the index set of grid points closest to $\boldsymbol{\vartheta}^{*}$ within the complete grid, and the corresponding grid angle set is denoted as $\boldsymbol{\vartheta}_{\mathcal{G}^{*}}$. Thus, $\theta_{g_k^*} - \Delta \theta < \theta_k^*$ and $\theta_{g_k^*} + \Delta \theta > \theta_k^*$. According to the aforementioned sign variation characteristic of the bias parameter, the bias vector around the nearest grid points satisfies
\begin{equation}
	\boldsymbol{\beta}(\boldsymbol{\vartheta}_{\mathcal{G}^*} - \Delta \boldsymbol{\vartheta}) > \mathbf{0}, \ \boldsymbol{\beta}(\boldsymbol{\vartheta}_{\mathcal{G}^*} + \Delta \boldsymbol{\vartheta}) < \mathbf{0},   \label{true_angle_property}
\end{equation}
where $\Delta \boldsymbol{\vartheta} = \Delta \theta \cdot \mathbf{1}_K$. A schematic illustration of this property is shown in Fig. \ref{spatial_angular_pseudo-derivative_schematic}.

Therefore, from the definition of the bias parameter, the SAPD Property naturally arises as a necessary property of sparse DOA solutions on the discretized spatial grid.

\begin{figure}[!t]
	\centering
	\includegraphics[width=3.2in]{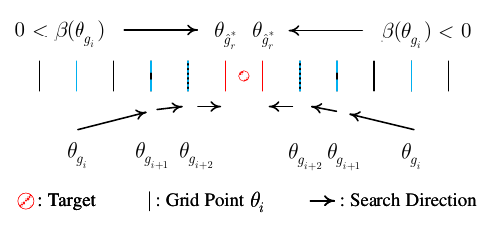}
	\caption{\centering{Graphical illustration of the SAPD property}}
	\label{spatial_angular_pseudo-derivative_schematic}
\end{figure}

For a given candidate index set $\mathcal{G}\subset [G]$ with cardinality $d$, the sparsity level is determined by the candidate solution and does not necessarily equal the actual number of incident sources $K$. To facilitate the subsequent discussion and derivation, we define the following function
\begin{align}
	\varepsilon(\boldsymbol{\vartheta}_\mathcal{G}) = \| \boldsymbol{y} - \mathbf{A}(\boldsymbol{\vartheta}_\mathcal{G})\boldsymbol{x}(\boldsymbol{\vartheta}_\mathcal{G}) \|_2 \label{recovery_error}, 
\end{align}
We first consider the sparse DOA estimation subproblem under a fixed sparsity level $d$.

\begin{theorem}[Existence of Local Minimizers]
	For any fixed cardinality $d$, the corresponding sparse DOA estimation subproblem
	\begin{equation}
		\min_{\boldsymbol{\vartheta}_{\mathcal{G}}}
		\ \varepsilon(\boldsymbol{\vartheta}_{\mathcal{G}}),
		\quad \text{s.t.} \ |\mathcal{G}| = d, \label{subspace_K_SD}
	\end{equation}
	possesses local minimizers.   \label{D_Local_Minima}
\end{theorem}

\textit{Proof:} See Appendix \ref{Proof:Th-1}. $\hfill\blacksquare$

Theorem \ref{D_Local_Minima} indicates that, given a fixed cardinality, the landscape of problem \eqref{subspace_K_SD} is analogous to that of the trimmed Lasso \cite{amir2021trimmed} and possesses a multitude of $d$-sparse local minima. 

Denote the grid index set corresponding to a local minimizer of the subproblem as $\hat{\mathcal{G}}^*$. According to the definition of the bias parameter introduced above, the corresponding bias parameter can also be constructed for any local minimizer. Similar to the case of $d=K$, the local minimizer also exhibits the SAPD Property
\begin{equation}
	\boldsymbol{\beta}(\boldsymbol{\vartheta}_{\hat{\mathcal{G}}}^*- \Delta \boldsymbol{\theta}) > \mathbf{0}, \ \boldsymbol{\beta}(\boldsymbol{\vartheta}_{\hat{\mathcal{G}}}^* + \Delta \boldsymbol{\theta}) < \mathbf{0}.   \label{extend_angle_property}
\end{equation}
Therefore, according to the definition of the bias parameter, the SAPD Property represents a necessary property exhibited by the local minima of the subproblem on the discretized spatial grid.

\textbf{Remark:} It should be noted that although the definition of the SAPD Property is revealed from the definition of the bias parameter, and the bias parameter has a similar mathematical expression to the parameter describing angular deviation in off-grid estimation, the two concepts address fundamentally different problems and serve different purposes. The SAPD Property characterizes the structural feature of sparse DOA solutions on the discretized spatial grid, which provides the basis for exploiting the low computational complexity of the candidate solution verification process in subsequent procedures. Specifically, it enables avoiding the involvement of the large-scale matrix computations in subsequent computations and prevents the exhaustive evaluation of a large number of unordered candidate combinations. In contrast, off-grid estimation is mainly used to improve DOA estimation accuracy. Therefore, the two approaches target fundamentally different problems and serve \textbf{completely different} purposes.

The above analysis reveals the SAPD Property exhibited by sparse DOA solutions on the discretized spatial grid. Based on this property, a new objective function is further formulated, and the corresponding optimization method is introduced in the following section.

\subsection{SAPD-Constrained Sparse DOA Optimization Function}

To encode the SAPD Property into an equality constraint, we define the function 
\begin{align}
	\boldsymbol{h}(\boldsymbol{\vartheta}_{\mathcal{G}}) = \operatorname{sgn}{ (\boldsymbol{\beta}(\boldsymbol{\vartheta}_{\mathcal{G}} - \Delta \boldsymbol{\vartheta})) } + \operatorname{sgn}{(\boldsymbol{\beta}(\boldsymbol{\vartheta}_{\mathcal{G}} + \Delta \boldsymbol{\vartheta}))}. \label{h}
\end{align}
By imposing $\boldsymbol{h}(\boldsymbol{\vartheta}_{\mathcal{G}}) = \mathbf{0}$, the resulting equality constraint is referred to as the SAPD Constraint. Based on the above analysis, the SAPD Constraint is incorporated into the conventional sparse DOA estimation formulation, which can be reformulated as
\begin{equation}
	\begin{aligned}
		\min \quad &\| \boldsymbol{x} \|_0
		\\
		\text{s.t.} \quad \varepsilon(\boldsymbol{\vartheta}_{\mathcal{G}}) < & \ \epsilon, 
		\ \boldsymbol{h}(\boldsymbol{\vartheta}_{\mathcal{G}}) =  \mathbf{0}.
	\end{aligned} \label{SAPD_Simple}
\end{equation} 
where $\epsilon$ specifies the recovery error tolerance and determines the sparsity level of the resulting solution. The formulation in \eqref{SAPD_Simple} is referred to as the SAPD-Constrained Sparse DOA Optimization problem. In the SAPD-Constrained Sparse DOA formulation, the recovery error tolerance $\epsilon$ mainly controls the automatic estimation of the number of incident sources. The proposed method is insensitive to the selection of $\epsilon$, and its influence on the estimation performance will be further discussed in the experimental section.

Due to the introduction of the SAPD Property and SAPD Constraint, although the optimization problem \eqref{SAPD_Simple} still adopts the conventional $\ell_0$-norm as the sparsity objective, its solving process no longer requires an unordered traversal over all feasible candidates. Instead, the spatial offset information provided by the bias parameter can be exploited to perform an ordered search among candidate solutions. Specifically, the sign of the bias parameter determines the search direction, while its magnitude determines the search distance.

Since the above search process is performed on the discretized spatial grid and the offsets between candidate solutions are represented by grid indices, the search step size belongs to the set of positive integers $\mathbb{Z}_{+}$. Therefore, in the practical search process, the absolute value of the bias parameter only needs to be converted into the corresponding grid offset step, which is achieved by applying the ceiling operation.

During the subsequent search process, the optimization problem \eqref{SAPD_Simple} is solved by sequentially evaluating the corresponding subproblems \eqref{D_Local_Minima}, starting from the minimum and most probable sparsity level and gradually increasing the sparsity level $d$ until $d=M-1$. Once a subproblem first yields a solution satisfying both the recovery error tolerance and the SAPD Constraint, the corresponding subproblem solution is regarded as the final solution of \eqref{SAPD_Simple}. Since the above search is performed on the discretized spatial grid, the obtained estimate is an on-grid solution, which can be further refined through off-grid refinement.

Therefore, compared with the original $\ell_0$-norm minimization problem, the main difference of the proposed SAPD-Constrained Sparse DOA formulation \eqref{SAPD_Simple} does not lie in modifying the sparse optimization objective, but rather in exploiting the structural relationship between the true incident angles and the discretized spatial grid revealed by the SAPD Property to provide a search order for the originally unordered feasible solution verification process. By introducing the SAPD Constraint and the grid offset estimation based on the bias parameter, the proposed method transforms the unordered verification of a large number of candidate solutions in conventional $\ell_0$-norm optimization into an ordered grid search process with directional and step-size information. Consequently, compared with unordered exhaustive verification, the proposed search strategy avoids the individual evaluation of a large number of unnecessary candidate combinations, thereby reducing the overall search complexity.

\subsection{The Approximate of Bias Parameter}

Since the bias parameter $\boldsymbol{\beta}_{\mathcal{G}}$ cannot be directly obtained from the measurements, it needs to be indirectly estimated based on the observation model. 

A first-order Taylor expansion is adopted to approximate the array manifold, thereby establishing an estimable model of the bias parameter. Accordingly, $\boldsymbol{y} \approx \mathbf{A}(\boldsymbol{\vartheta}_{\mathcal{G}})\boldsymbol{x}_{\mathcal{G}} + \mathbf{B}(\boldsymbol{\vartheta}_{\mathcal{G}})\operatorname{diag}(\boldsymbol{\beta}_{\mathcal{G}})\boldsymbol{x}_{\mathcal{G}} + \boldsymbol{n}$ , where $\mathbf{B}(\boldsymbol{\vartheta}) = [\boldsymbol{b}(\theta_{g_1}), \dots, \boldsymbol{b}(\theta_{g_d})]$ with $\boldsymbol{b}(\theta_g)=\frac{\partial \boldsymbol{a}(\theta_g)}{\partial \theta_g}$, and $\boldsymbol{\beta}_{\mathcal{G}} \in \mathbb R^d$ denotes the bias parameter vector associated with the given grid points indexed by $\mathcal G$ with cardinality $d$. 

After obtaining the corresponding sparse coefficient vector 
$\boldsymbol{x}_{\mathcal G}$ via \eqref{x_theta_function}, 
the following derivation follows the idea in \cite{SC_Bilinear}. Using the identity $\operatorname{diag}(\boldsymbol{\beta}_{\mathcal{G}})\boldsymbol{x}_{\mathcal{G}} = \operatorname{diag}(\boldsymbol{x}_{\mathcal{G}})\boldsymbol{\beta}_{\mathcal{G}}$, the observation model is reformulated as  $\boldsymbol{y} = \mathbf{A}(\boldsymbol{\vartheta}_{\mathcal{G}}) \boldsymbol{x}_{\mathcal{G}} + \mathbf{B}({ \boldsymbol{\vartheta}_{\mathcal{G}} }) \operatorname{diag}(\boldsymbol{x}_{\mathcal{G}}) { \boldsymbol{\beta}_{\mathcal{G}} } + \boldsymbol{n}$. Therefore, the estimation of the bias parameter $\boldsymbol{\beta}_{\mathcal{G}}$ associated with
$\mathcal G$ can be formulated as the following real-valued constrained
least-squares problem
\begin{equation}
	\mathop{\arg\min}\limits_{\boldsymbol{\beta}_{\mathcal{G}} \in \mathbb{R}^d} \| \tilde{\boldsymbol{y}} - \tilde{\mathbf{B}}({ \boldsymbol{\vartheta}_{\mathcal{G}} }) \boldsymbol{\beta}_{\mathcal{G}} \|_2^2, \label{bias_theta_function}
\end{equation}
where $\tilde{\boldsymbol{y}} = \boldsymbol{y} -\mathbf{A}(\boldsymbol{\vartheta}_{\mathcal{G}})\boldsymbol{x}_{\mathcal{G}}$ and $\tilde{\mathbf{B}}({ \boldsymbol{\vartheta}_{\mathcal{G}} }) = \mathbf{B}({ \boldsymbol{\vartheta}_{\mathcal{G}} }) \operatorname{diag}(\boldsymbol{x}_{\mathcal{G}})$. 
Under the real-valued constraint $\boldsymbol{\beta}_{\mathcal{G}} \in \mathbb{R}^{\vert \mathcal{G} \vert}$, the closed-form solution to \eqref{bias_theta_function} associated with $\mathcal G$ is given by
\begin{equation}
	\begin{aligned}
		\boldsymbol{\beta}_{\mathcal{G}} &= \boldsymbol{\beta}(\boldsymbol{\vartheta}_{\mathcal{G}}) 
		\\
		&=  \{\operatorname{Re}(\tilde{\mathbf{B}}({ \boldsymbol{\vartheta}_{\mathcal{G}} })^H \tilde{\mathbf{B}}({ \boldsymbol{\vartheta}_{\mathcal{G}} }))\}^{-1} \cdot \operatorname{Re}\{ \tilde{\mathbf{B}}({ \boldsymbol{\vartheta}_{\mathcal{G}} })^H \tilde{\boldsymbol{y}} \}.
	\end{aligned} \label{spatial_angular_pseudo_dev}
\end{equation}
The formulation in \eqref{spatial_angular_pseudo_dev} has a similar representation to the iterative off-grid parameter estimation formulation in \cite{Off-grid_Chen}, where both are derived based on a first-order Taylor expansion. However, the purpose of utilizing this estimation result in this work is fundamentally different from that of conventional off-grid refinement.

During the verification of the SAPD Constraint, the bias parameter associated with a given index set $\mathcal{G}$ needs to be calculated. However, since the grid points in $\mathcal{G}$ do not necessarily correspond to the grid points closest to the true incident angles in the complete discretized spatial grid $\boldsymbol{\vartheta}$, the exact values of the bias parameters are generally difficult to obtain directly. On the other hand, since the SAPD Constraint mainly exploits the sign information of the bias parameter and converts its magnitude into the search step size on the discretized grid, an exact estimation of the bias parameter is not required. Instead, only an approximate estimation that can capture its spatial variation characteristics is needed. Specifically, for a given index set $\mathcal{G}$, the corresponding sparse coefficient vector $\boldsymbol{x}_{\mathcal{G}}$ is first obtained via \eqref{x_theta_function}, and the bias parameter $\boldsymbol{\beta}_{\mathcal{G}}$ is subsequently estimated through \eqref{spatial_angular_pseudo_dev}.

Therefore, during the candidate solution verification process, the SAPD Constraint only introduces one additional least-squares estimation compared with the conventional exhaustive verification process, resulting in a limited computational overhead. Consequently, the overall computational complexity of the verification process remains $O(Md^2)$, while avoiding the involvement of the full overcomplete dictionary or covariance-like matrices in subsequent computations.

Nevertheless, since problem \eqref{SAPD_Simple} remains a nonconvex optimization problem, the quality of initialization plays a critical role in the success of the above procedure. The solution strategy for problem \eqref{SAPD_Simple} is therefore developed in the next section.

\section{SAPD Search Algorithm} \label{Sec4}

The proposed method for solving problem \eqref{SAPD_Simple} is referred to as the SAPD Search Algorithm. For convenience, a given index set $\mathcal{G}$ is referred to as the verification index set. The corresponding fixed cardinality $\vert \mathcal{G} \vert = d$ is referred to as the verification sparsity level. Furthermore, the grid point corresponding to each element in $\mathcal{G}$ is referred to as a verification grid point. 

The SAPD Search Algorithm consists of three components which are Spatial Spectrum Information Extraction, the Main Loop, and DOA Refinement. The overall framework of the proposed algorithm is illustrated in Fig. \ref{Overview_of_SAPD_Search_Scheme_Procedure}.

Spatial Spectrum Information Extraction consists of two procedures, namely Global Initialization and Global Compensation Prioritization. This component provides the initialization information for problem \eqref{SAPD_Simple} and its associated fixed-sparsity subproblems \eqref{subspace_K_SD}.

The Main Loop starts from the smallest and most plausible verification sparsity level determined by Global Initialization. Using the initial verification index set $\mathcal{G}$, the corresponding sparse DOA estimation subproblem \eqref{subspace_K_SD} is solved via SAPD-Guided Search. If the obtained solution does not satisfy the prescribed recovery error tolerance $\epsilon$, the verification sparsity level is increased. Subsequently, Global Compensation Prioritization provides the initialization information for the new fixed-sparsity subproblem \eqref{subspace_K_SD}, and Subproblem Reinitialization is performed to update the verification index set $\mathcal{G}$. SAPD-Guided Search is then executed again. This process continues until a solution satisfying both the SAPD Constraints and the prescribed recovery error tolerance $\epsilon$ is obtained for the first time. The output of the Main Loop is the corresponding on-grid solution.

Finally, the obtained on-grid solution is used to initialize the estimation of the corresponding off-grid bias, thereby refining the DOA estimates.

\begin{figure}[!t]
	\centering
	\includegraphics[width=3.2in]{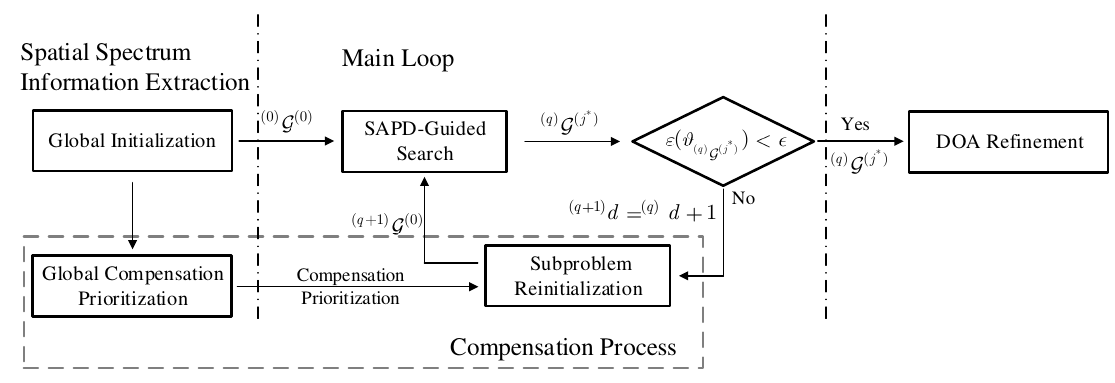}
	\caption{\centering{Overview of SAPD Search Scheme Procedure}}
	\label{Overview_of_SAPD_Search_Scheme_Procedure}
\end{figure}

SAPD-Guided Search only updates the elements of $\mathcal{G}$, whereas Subproblem Reinitialization updates both the cardinality and the elements of $\mathcal{G}$. To distinguish these two types of updates, a two-level superscript notation is adopted throughout the search process. Specifically, the left superscript $(q)$ denotes the iteration number of the Main Loop, i.e., the number of updates of $\mathcal{G}$ performed by Subproblem Reinitialization. The verification sparsity level at the $q$-th Main Loop iteration is denoted by ${}^{(q)}d$. The right superscript $(j)$ denotes the $j$-th search iteration of SAPD-Guided Search under the current verification sparsity level ${}^{(q)}d$. Accordingly, ${}^{(q)}\mathcal{G}^{(j)}$ represents the verification index set corresponding to the $j$-th search iteration of SAPD-Guided Search during the $q$-th iteration of the Main Loop under the current verification sparsity level ${}^{(q)}d$. Therefore, the proposed method estimates the number of incident sources automatically and does not require prior knowledge of the number of sources.

The remainder of this section presents the individual components of the SAPD Search Algorithm in sequence, namely Global Initialization, SAPD-Guided Search, Global Compensation Prioritization, and Subproblem Reinitialization.

\subsection{Global Initialization based Spatial Spectrum}

Since the optimization problem \eqref{SAPD_Simple} and the sequence of fixed-sparsity subproblems \eqref{subspace_K_SD} are inherently nonconvex, an inappropriate initialization strategy may not only increase the number of search iterations but also cause the final solution to become trapped in a local optimum. Therefore, an effective initialization strategy is of critical importance for solving problem \eqref{SAPD_Simple} and its sequence of fixed-sparsity subproblems \eqref{subspace_K_SD}. 

In this subsection, we introduce the Global Initialization procedure. This procedure exploits the spatial spectrum obtained via Bartlett beamforming to provide the initialization information required by the proposed SAPD Search Algorithm and generates the initial verification index set  ${}^{(0)}\mathcal{G}^{(0)}$ for the overall optimization problem. Its corresponding verification sparsity level ${}^{(0)}d$ is the smallest and most plausible initial sparsity level throughout the entire search process.

Bartlett beamforming is then employed to estimate the power at each grid point, denoted by  $P_{\text{s}}(\theta_g)$ for $\theta_g \in \boldsymbol{\vartheta}$, $g = 1, \dots, G$. In the single-snapshot case, the spatial spectrum can be expressed as \cite{96-Twodecades}
\begin{align}
	P_{\text{s}}(\theta_g) = \vert \boldsymbol{a}^H(\theta_g) \, \boldsymbol{y} \vert^2. \label{ss_bb}
\end{align}
Let $P_{\text{s}}(\boldsymbol{\vartheta}) \in \mathbb{R}^{G}$ denote the spatial-spectrum vector over all grid points. The normalized spatial spectrum is defined as $\tilde{P}_\text{s}(\boldsymbol{\vartheta}) = \tfrac{P_\text{s}(\boldsymbol{\vartheta})}{ \| P_\text{s}(\boldsymbol{\vartheta}) \|_{\infty} }$. To estimate the spectral noise floor $n_f$, we first identify the subset $\mathcal{N}$ of low-power samples in the spatial spectrum, which is defined as
\begin{equation}
	\begin{aligned}
		\mathcal{N} = \{ P_{\text{s}}(\theta_{m_i}) \mid
		\tilde{P}_{\text{s}}(\theta_{m_i}) < \eta_n, \ 
		m_i \in \mathcal{M}, \  \mathcal{M} \subset [G]
		\}
	\end{aligned}
\end{equation}
where $\eta_n \in \mathbb{R}$ denotes a threshold used to distinguish the background noise floor from the local minima located between adjacent spectral peaks. Then, the spectral noise floor $n_f$ is defined as the mean power of the grid points within the set $\mathcal{N}$. The regions corresponding to the set $\mathcal{N}$ are generally assumed to contain no incident signals.

According to the conventional interpretation in spatial spectrum analysis, the mainlobe in the spatial spectrum is typically assumed to be generated by a single incident signal. Therefore, when adjacent mainlobes become merged, it is generally considered that the current spatial spectrum can no longer effectively resolve the corresponding incident signals, implying estimation failure. However, such treatment often neglects the substantial amount of angular information still contained in the merged spatial spectrum. Existing literature provides relatively limited characterization of the structural properties exhibited by spatial spectra in the presence of mainlobe merging. To avoid ambiguity in the subsequent discussion, it is therefore necessary to further introduce several new concepts capable of describing such spatial spectrum structures, thereby enabling a clearer characterization of the incident-angle information contained in the spatial spectrum.

The set of detected spectral peaks $\mathcal{P}$ is defined as
\begin{equation}
	\mathcal{P} = \{ p_i \mid p_i = \tilde{P}_{\text{s}}(\theta_i), \  p_i  > \eta_t, \  \theta_i \in \boldsymbol{\theta}, \  i \in [D] \}, 
\end{equation} 
where $\eta_t > \eta_n$ denotes the peak detection threshold. Owing to the Rayleigh resolution limit and peak merging effects, the number of resolved spectral peaks $D$ satisfies $D \le K$. The half-power level corresponding to the peak $p_i$ is defined as $P_{h_i} = \tfrac{1}{2} (p_i - n_f)$.  

Let $\bar{\theta}_i$ denote the angular position corresponding to the spectral peak $p_i$. The angle of the first grid point located to the left of $\bar{\theta}_{i}$ satisfying $P_{\text{s}}(\bar{\theta}_f^{l_i}) < P_{h_i}$ is denoted by $\bar{\theta}_f^{l_i}$, where $\bar{\theta}_f^{l_i} < \bar{\theta}_i$. The nearest local $\bar{\theta}_m^{l_i}$ minimum located to the left of $\bar{\theta}_i$, where $\bar{\theta}_m^{l_i} < \bar{\theta}_i$. The generalized half power point $\theta_{h}^{l_i}$ in the left of $\bar{\theta}_{i}$ is defined as
\begin{equation}
	\theta_{h}^{l_i} = \min ( \bar{\theta}_i - \bar{\theta}_f^{l_i}, \  \bar{\theta}_i - \bar{\theta}_m^{l_i} ).
\end{equation}
Similarly, the angle of the first grid point located to the right of $\bar{\theta}_{i}$ satisfying $P_{\text{s}}(\bar{\theta}_f^{r_i}) < P_{h_i}$ is denoted by $\bar{\theta}_f^{r_i}$, where $\bar{\theta}_f^{r_i} > \bar{\theta}_i$. The nearest local minimum $\bar{\theta}_m^{r_i}$ located to the right of $\bar{\theta}_i$, where $\bar{\theta}_m^{r_i} > \bar{\theta}_i$. The generalized half power point $\theta_{h}^{r_i}$ in the right of $\bar{\theta}_{i}$ is defined as
\begin{equation}
	\theta_{h}^{r_i} = \min ( \bar{\theta}_f^{r_i} - \bar{\theta}_i, \  \bar{\theta}_m^{r_i} -  \bar{\theta}_i).
\end{equation}

To facilitate the characterization of unresolved mainlobe structures, the concept of the Generalized Mainlobe Region (GMR) is introduced.
\begin{definition}[Generalized Mainlobe Region, GMR] 
	For any spectral peak $p_i \in \mathcal{P}$, the corresponding generalized mainlobe region (GMR) is defined as the local spatial spectral region determined by the boundary angles
	\begin{equation}
		\mathcal{B}_i = [\theta_{h}^{l_i}, \theta_{h}^{r_i}].
	\end{equation}
\end{definition}

Based on the definition of GMR, the beamwidth of the GMR is called the generalized beamwidth, which is defined as
\begin{equation}	
	L_{B_i} = \theta_{h}^{r_i} - \theta_{h}^{l_i}. \label{beam_side_length}
\end{equation}

Let $\bar{\theta}_0 = \theta_1$ and $\bar{\theta}_{D+1} = \theta_G$, where $\theta_1$ and $\theta_G$ denote the starting and ending grid points of the spatial grid $\boldsymbol{\vartheta}$, respectively. The inter-peak region is defined as $\bar{\mathcal{V}}_i = [\bar{\theta}_i, \bar{\theta}_{i+1}]$, $i \in {0, 1, \dots, D}$. Let $P_{\mathcal{V}_i} = \min_{\theta \in \bar{\mathcal{V}}_i} P_\text{s} (\theta)$ denote the minimum spectral power within the $i$-th inter-peak region. The set of minimum spectral power values corresponding to all inter-peak regions is defined as
\begin{equation}
	\mathcal{M} = \{ P_{\mathcal{V}_0}, \ P_{\mathcal{V}_1}, \ \dots, \ P_{\mathcal{V}_{D}} \}.
\end{equation}

\begin{figure}[!t]
	\centering
	\includegraphics[width=3.2in]{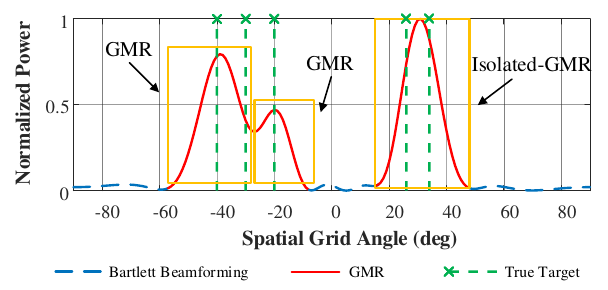}
	\caption{\centering{Different Types of GMR}}
	\label{DGMR}
\end{figure}

\begin{definition}[Valley Region, VR] 
	For each minimum spectral power value $P_{\mathcal{V}_i}$, the corresponding valley region (VR) is defined as the local spatial spectral region determined by the boundary angles
	\begin{equation}
		\mathcal{V}_i = [\theta_{v}^{l_i}, \theta_{v}^{r_i}], \  i = {0, 1, \dots, D},
	\end{equation}
	where $\theta_v^{l_i}$ and $\theta_v^{r_i}$ denote the left and right boundaries of the corresponding valley region, respectively.
\end{definition}

The left boundary $\theta_v^{l_i}$ of the valley region is determined as follows. For $i=0$, the left boundary is set to the starting grid point $\theta_1$. Otherwise, if the generalized right half-power point $\theta_h^{r_i}$ corresponds to the exact half-power point $\bar{\theta}_f^{r_i}$, the left valley boundary is defined as $\theta_h^{r_i} = \bar{\theta}_f^{r_i}$. If $\theta_h^{r_i}$ corresponds to the nearest local minimum $\bar{\theta}_m^{r_i}$ located to the right of the peak $\bar{\theta}_i$, the left valley boundary is instead defined as $\theta_{v}^{l_i} = \bar{\theta}_i$.

Similarly, if the generalized left half-power point $\theta_h^{l_{i+1}}$ corresponds to the exact half-power point $\bar{\theta}_m^{l_{i+1}}$, the right valley boundary is defined as $\theta_v^{r_i} = \theta_h^{l_{i+1}}$. Otherwise, if $\theta_h^{l_{i+1}}$ corresponds to the nearest local minimum $\bar{\theta}_m^{l_{i+1}}$ located to the left of the peak $\bar{\theta}_{i+1}$, the right valley boundary is defined as $\theta_h^{r_i} = \bar{\theta}_{i+1}$. For $i = D$, the right boundary is set to the ending grid point $\theta_G$.

Based on the definition of the valley region (VR), two additional concepts are introduced to further characterize the structural properties of the spatial spectrum. If $P_{\mathcal{V}_i} > n_f $, the corresponding VR is referred to as an \textbf{unresolved valley region} (UVR), indicating that the valley region may still contain unresolved signal information. Conversely, if $P_{\mathcal{V}_i} < n_f $, the corresponding VR is referred to as a \textbf{resolved valley region} (RVR), indicating that the corresponding spectral components are sufficiently separated and the valley region is unlikely to contain additional signal information.

Based on the definitions of the GMR and VR, the structural characteristics of the spatial spectrum can be more precisely described. In particular, if the valley regions
$\mathcal{V}_{i-1}$ and  $\mathcal{V}_i$ adjacent to a GMR $\mathcal{B}_i$ are both resolved valley regions (RVRs), the GMR $\mathcal{B}_i$ is considered to be spectrally separated from its neighboring GMRs. Such a structure is referred to as an \textbf{isolated generalized mainlobe region}  (Isolated-GMR).

\begin{figure}[!t]
	\centering
	\includegraphics[width=3.1in]{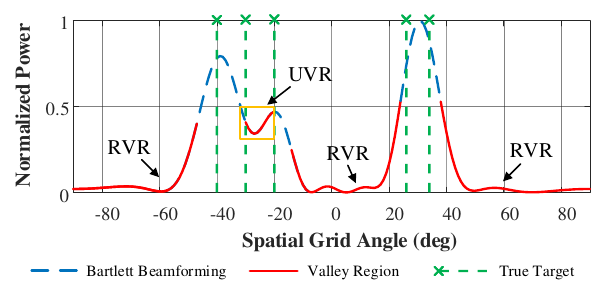}
	\caption{\centering{Different types of VR}}
	\label{DVR}
\end{figure}

\textbf{Remarks}: It should be emphasized that the concept of isolated-GMR characterizes only the separability between adjacent GMR structures, rather than the number of incident sources contained within a GMR. Therefore, an isolated-GMR may still contain multiple unresolved incident sources caused by mainlobe merging. In other words, the isolated property only indicates that the corresponding GMR is sufficiently separated from neighboring spectral structures, such that its local width characteristics can be reliably analyzed without significant interference from adjacent GMRs.

Based on the above structural definitions, several observable angular characteristics can be directly inferred from the spatial spectrum under moderate-to-high SNR conditions.

\textbf{Observation 1}: First, each GMR is associated with at least one incident source.

\textbf{Observation 2}: Second, for an isolated-GMR, the influence of neighboring GMRs can be neglected due to the existence of RVRs on both sides. Under the single-source assumption, the width of the corresponding mainlobe is expected to remain within the nominal beamwidth range. Therefore, if $L_{\mathcal{B}_i} > \theta_{3dB} + \delta$, where $\theta_{3dB}$ denotes the nominal beamwidth and $\delta$ is a compensation term accounting for noise perturbation and discretization effects, the corresponding isolated-GMR can be inferred to contain at least two incident sources.

The compensation term $\delta$ is an empirical parameter. In millimeter-wave radar DOA estimation, under commonly used angular grid resolutions, the discretization error is typically within $[1^\circ, \, 2^\circ]$, while noise perturbation usually manifests as a slight broadening of the mainlobe. Therefore, $\delta$ is uniformly set to $2^\circ$ throughout this paper. This choice ensures reliable multi-source detection while avoiding unnecessary repeated detections caused by an excessively large compensation range.

\textbf{Observation 3}: Third, unresolved valley regions (UVRs) indicate that adjacent GMRs are not fully separated, implying that unresolved incident sources may still exist within the corresponding spectral interval.

Based on the observable angular information inferred from the spatial spectrum structures described above, we now introduce the initialization procedure of the proposed SAPD Search Scheme.

\textbf{Initialization ${}^{(0)}\mathcal{G}^{(0)}$}: By exploiting the structural information provided by the spatial spectrum, both the number of incident sources and the initial locations for the subsequent search process can be initialized. However, it should be emphasized that if the number of initialized angles exceeds the true number of incident sources, the algorithm will directly produce erroneous estimation results.

Therefore, a conservative initialization strategy is adopted in this work. Specifically, the initialization procedure preferentially selects the spatial spectral regions that are most likely to contain incident sources, and initializes both the number of initial incident angles and the starting locations for the subsequent search process based on these regions.

First, a preliminary detection procedure is applied to the spatial spectrum to estimate the noise floor $n_f$, identify the set of spectral peaks $\mathcal{P}$, and extract the corresponding GMRs $\mathcal{B}_i$. Subsequently, the VRs, UVRs, RVRs, and isolated-GMRs are determined. 

Based on the above observations, it is assumed that incident angles are most likely located within the GMRs. Therefore, in the Global Initialization step, both the initial angle estimates and the corresponding initial sparsity level for the subsequent SAPD search are determined only within the GMRs.

According to Observation 2, the initialization procedure for an isolated-GMR $\mathcal{B}_i$ is performed as follows. If $L_{\mathcal{B}i} < \theta_{3dB} + \delta$, the initialized angle set $\mathcal{I}_i$ contains only one angle, namely the peak location of the corresponding GMR, i.e.,
\begin{equation}
	\mathcal{I}_i = \{ \bar{\theta}_i \}.
\end{equation}
Otherwise, if $L_{\mathcal{B}i} > \theta_{3dB} + \delta$, two angles are initialized within the corresponding GMR, i.e.,
\begin{equation}
	\mathcal{I}_i = \{ \bar{\theta}_{\mathcal{I}_i}^1, \  \bar{\theta}_{\mathcal{I}_i}^2  \}, 
\end{equation}
where
\begin{equation}
	\bar{\theta}_{\mathcal{I}_i}^1 = \theta_h^{l_i} + \lceil \frac{\bar{\theta}_i - \theta_h^{l_i}  }{2} \rceil, \   \bar{\theta}_{\mathcal{I}_i}^2 = \bar{\theta}_i + \lceil \frac{ \theta_h^{r_i} -  \bar{\theta}_i }{2} \rceil.
\end{equation}

If the GMR $\mathcal{B}_i$ is not an isolated-GMR, it may contain a more complicated peak-merging structure. Therefore, according to Observation 1, only one angle is initialized within the corresponding GMR during the initialization stage, namely the angle associated with its peak location, i.e.,
\begin{equation}
	\mathcal{I}_i = \{ \bar{\theta}_i \}.
\end{equation}

The overall initialized angle set is then defined as the union of all initialized angle sets, i.e.,
\begin{equation}
	\mathcal{I} = \bigcup_{i = 1}^{\vert \mathcal{P} \vert} \mathcal{I}_i.
\end{equation}
\textbf{The first verification index set ${}^{(0)}\mathcal{G}^{(0)}$ is then obtained as the set of grid indices corresponding to the angles in $\mathcal{I}$.}

The steps of the initialization process are outlined in Algorithm \ref{alg:1}.
\begin{algorithm}
	\caption{Global Initialization via Spatial Spectrum}
	\label{alg:1}
	\begin{algorithmic}[1]
		\Require Observation data $\boldsymbol{y}\in\mathbb{C}^{M}$
		\Ensure Initial verification index set ${}^{(0)}\mathcal{G}^{(0)}$
		
		\State Compute the spatial spectrum $P_{\mathrm{s}}(\boldsymbol{\vartheta})$ using \eqref{ss_bb}
		
		\State Estimate the noise floor $n_f$, identify the valid peak set $\mathcal{P}$, and determine the GMRs, VRs, UVRs, RVRs, and isolated-GMRs
		
		\State Initialize each GMR according to Observations 1 and 2, and obtain the initialized angle set $\mathcal{I}$ and the corresponding verification index set ${}^{(0)}\mathcal{G}^{(0)}$
		
		\State \Return ${}^{(0)}\mathcal{G}^{(0)}$
	\end{algorithmic}
\end{algorithm}

\subsection{SAPD-Guided Search Step}

After Global Initialization or Subproblem Reinitialization, SAPD-Guided Search is employed to solve the subproblem \eqref{subspace_K_SD} associated with the current verification sparsity level ${}^{(q)}d$. Since the subproblem \eqref{subspace_K_SD} admits local minimizers and the SAPD Property of the sparse DOA solutions provides information regarding the relative position of the current solution with respect to a local minimizer, the bias parameter $\boldsymbol{\beta}(\boldsymbol{\vartheta}_{{}^{(q)}\mathcal{G}^{(j)} })$ obtained from \eqref{spatial_angular_pseudo_dev} can be utilized as a search indicator to update the elements of the verification index set ${}^{(q)}\mathcal{G}^{(j)}$. Motivated by this observation, SAPD-Guided Search transforms the solution of subproblem \eqref{subspace_K_SD} into a discrete search process on the spatial grid and progressively searches for a solution satisfying the SAPD Constraint.

During SAPD-Guided Search, the verification sparsity level ${}^{(q)}d$ remains unchanged. Therefore, only the values of the elements within the verification index set need to be updated. For notational simplicity, the left superscript ${ (q) }$ is omitted in the following discussion, and the verification index set corresponding to the $j$-th search step is denoted by $\mathcal{G}^{(j)} $. The detailed search process is described as follows.

1) Search Step: During the $j-1$-th search, the sparse coefficient vector $\boldsymbol{x}_{\mathcal{G}^{(j-1)}}$ and the corresponding bias parameter $\boldsymbol{\beta}_{\mathcal{G}^{(j - 1)}}$ are obtained from the current angle index set $\mathcal{G}^{(j-1)}$ through \eqref{x_theta_function} and \eqref{spatial_angular_pseudo_dev}, respectively. Subsequently, based on the SAPD property, the search step corresponding to the $j$-th search, denoted by $\boldsymbol{S}_j \in \mathbb{R}^k$, is defined as
\begin{equation}
	\boldsymbol{S}^{(j)} =  \operatorname{sgn}(\boldsymbol{\beta}^{(j - 1)}) \cdot \left\lceil \vert \boldsymbol{\beta}^{(j - 1)} \vert \right\rceil \label{search_step_size}
\end{equation}
Subsequently, the current angle index set $\mathcal{G}^{(j - 1)}$ is updated to obtain the angle index set corresponding to the $j$-th search, i.e.,
\begin{equation}
	\mathcal{G}^{(j)} = \mathcal{G}^{(j-1)} + S^{(j)}. \label{update_grid_index}
\end{equation}
After obtaining the angle index set $\mathcal{G}^{(j)}$ for the $j$-th search, the sparse coefficient vector $\boldsymbol{x}_{\mathcal{G}^{(j)}}$ is re-estimated using \eqref{x_theta_function}. Then, the corresponding bias parameter $\boldsymbol{\beta}_{\mathcal{G}^{(j)}}$ is computed according to \eqref{spatial_angular_pseudo_dev}. Subsequently, the current search result is examined to determine whether the search termination condition is satisfied. If the termination condition is not met, the above search procedure is repeated.

2) Termination criterion: Since SAPD-Guided Search is essentially a discrete search process, conventional convergence criteria used in optimization problems are not adopted. Instead, a search termination criterion based on the oscillatory behavior of the angle index set is employed. Furthermore, such oscillatory behavior can be characterized through the variation pattern of consecutive search steps. To this end, a decision function $\boldsymbol{C}(j)$ is defined as
\begin{equation}
	\boldsymbol{C}(j) = \boldsymbol{S}^{(j-2)} + \boldsymbol{S}^{(j-1)} + \boldsymbol{S}^{(j)}.
\end{equation}

According to the definition of the search step $\boldsymbol{S}_j$, the term $\left\lceil \vert \boldsymbol{\beta}^{(j-1)} \vert \right\rceil$ corresponds to the ceiling operation applied to the magnitude of the bias parameter. Therefore, when the search enters the region corresponding to the SAPD Constraints, each component of the search step degenerates to $\pm 1$. If 
\begin{equation}
	\sum_{i = 1}^k \vert \boldsymbol{C}(j)_i \vert = k, \label{terminator}
\end{equation}
then, since each component of the search step can only take values from $\pm 1$, the above condition implies that three consecutive search steps exhibit an alternating sign pattern at each corresponding component. Consequently,
\begin{equation}
	\boldsymbol{S}_{j-2} = - \boldsymbol{S}_{j-1} = \boldsymbol{S}_{j}.
\end{equation}
This indicates that the search process repeatedly moves back and forth between two neighboring discrete angular locations for every potential DOA, yielding $\mathcal{G}^{(j-2)} = \mathcal{G}^{(j)}$. In other words, the search process no longer generates a new angle index set and instead enters a periodic oscillation between two adjacent grid locations. According to the definition of the SAPD Constraints, such an oscillatory state corresponds to a search result satisfying the SAPD Constraints, i.e., $h(\boldsymbol{\theta}_{\mathcal{G}^{(j)}}) = \mathbf{0}$. Therefore, the current search is regarded as having reached the termination condition, and the final search iteration is denoted by $\tilde{j}$.

3) The output of SAPD-Guided Search: After the search termination condition is satisfied, the search process enters a discrete oscillation state. Therefore, among the two oscillatory verification angular index sets that trigger the termination condition, namely $\mathcal{G}^{(\tilde{j} - 1)}$ and $\mathcal{G}^{(\tilde{j})}$, the angular index set corresponding to the minimum reconstruction error is selected as the output of the current SAPD Search, i.e.,
\begin{equation}
	\mathcal{G}^{(j*)} = \mathop{\arg\min}\limits_{\mathcal{G} \in \{ \mathcal{G}^{(\tilde{j} - 1)}, \mathcal{G}^{(\tilde{j})} \} }  \{ \varepsilon( \boldsymbol{\theta}_{\mathcal{G}^{(\tilde{j} - 1)}} ),  \varepsilon( \boldsymbol{\theta}_{\mathcal{G}^{(\tilde{j})}} ) \}. \label{SAPD_Search_Output}
\end{equation}
where $j^* \in \{ \tilde{j} - 1, \ , \tilde{j} \}$ denotes the search iteration associated with the selected output angular index set.

If $\varepsilon( \boldsymbol{\theta}_{\mathcal{G}^{(j*)}} ) > \epsilon$, the current verification sparsity level is insufficient to satisfy the recovery error tolerance $\epsilon$. In this case, the flag variable is set to $F_s = 1$, indicating that the SAPD Search Algorithm remains in the Main Loop, where the verification sparsity level is increased sequentially until an on-grid solution satisfying both the recovery error tolerance $\epsilon$ and the SAPD Constraints is obtained.

Otherwise, if $\varepsilon(\boldsymbol{\theta}_{\mathcal{G}^{(j*)}}) < \epsilon$,  the current verification index set ${G}^{(j*)}$ is accepted as the on-grid solution, $F_s$ is set to 0, and the Main Loop is terminated. Based on the above derivation, the complete procedure of the SAPD Search Step is summarized in Algorithm \ref{SAPD_Search_Step}.

\begin{algorithm}
	\caption{SAPD Search Step}
	\label{SAPD_Search_Step}
	
	\begin{algorithmic}[1]
		
		\Require Observation data $\boldsymbol{y}\in\mathbb{C}^{M}$,
		initialization angular index set $\mathcal{G}^{(0)}$,
		maximum search number $J$, and $F_s = 1$
		
		\Ensure The angular index sets ${G}^{(j*)}$, and the flag variable $F_s$
		
		\For{$j = 1$ to $J$}
		
		\State Calculate the sparse vector coefficient
		$\boldsymbol{x}_{\mathcal{G}^{(j - 1)}}$
		via \eqref{x_theta_function}
		
		\State Calculate the grid bias
		$\boldsymbol{\beta}_{\mathcal{G}^{(j - 1)}}$
		via \eqref{spatial_angular_pseudo_dev}
		
		\If{$j \ge 3$ and $\boldsymbol{C}(j)$ satisfies \eqref{terminator} }
		
		\State Break
		
		\EndIf
		
		\State Calculate the search step
		$\boldsymbol{S}^{(j)}$
		via \eqref{search_step_size}
		
		\State Update the angular index set $\mathcal{G}^{(j)}$ via \eqref{update_grid_index}.
		
		\EndFor
		
		\State Determine $\mathcal{G}^{(j^*)}$ via \eqref{SAPD_Search_Output}
		
		\If{$\varepsilon(\boldsymbol{\theta}_{\mathcal{G}^{(j^*)}})\le \epsilon$}
		
		\State Set $F_s = 0$
		
		\State \Return ${G}^{(j*)}$ and $F_s$.
		
		\Else
		
		\State Set $F_s = 1$
		
		\State \Return \Return $F_s$
		
		\EndIf
		
	\end{algorithmic}
\end{algorithm}

It is worth noting that the SAPD Search exploits the SAPD property established in this paper, enabling the search process to be performed directly in the discrete angular space through local search, rather than exhaustively traversing all possible angular combinations. Therefore, compared with conventional exhaustive search methods, the SAPD Search effectively avoids the exponential growth of the number of angular combinations as the sparsity level under verification increases, thereby significantly reducing the search space. 

Furthermore, in each search step of the SAPD Search, only two least-squares subproblems associated with the current angular index set need to be solved. These computations are performed based solely on the array manifold matrix corresponding to the current angular index set. Since the sparsity level under verification always satisfies $k < M$, the dimension of the array manifold matrix involved in the computation is at most $M \times (M - 1)$, and can be as small as $M \times 1$. Consequently, the SAPD Search always operates on low-dimensional array manifold matrices throughout the entire search procedure, without requiring matrix inversion or large-scale matrix multiplications on overcomplete dictionaries.

In addition, the SAPD Search is fundamentally different from conventional grid refinement and dynamic grid methods in sparse DOA estimation. Although these methods continuously refine or dynamically update grid locations during the iterative process, their optimization procedures are still built upon an overcomplete dictionary framework. Consequently, the corresponding sensing matrix remains an overcomplete dictionary of size $M \times G^\prime$, where $G^\prime > M$ denotes the dimension of the dynamically updated grid. Therefore, their computational procedures remain dependent on high-dimensional overcomplete dictionaries. In contrast, the SAPD Search completely eliminates the reliance on the iterative optimization framework based on overcomplete dictionaries. During the search process, the angular search at the current sparsity level under verification is performed solely using the array manifold matrix constructed from the current candidate angles, without requiring a large-scale overcomplete dictionary. This constitutes one of the fundamental differences between the proposed method and existing grid refinement and dynamic grid approaches. 

Since the entire search procedure is always carried out on low-dimensional array manifold matrices, the SAPD Search is able to effectively control the computational complexity while maintaining its search capability.

\subsection{Global Compensation Prioritization and Subproblem Reinitialization Step}

When $\varepsilon( \boldsymbol{\theta}_{{}^{(k)}\mathcal{G}^{(j*)}} ) > \epsilon$, the current verification sparsity level ${}^{(q)}d$ is considered insufficient to satisfy the minimum support sparsity required for sparse DOA estimation. Since SAPD-Guided Search cannot modify the verification sparsity level, the algorithm returns to the Main Loop, where the verification sparsity level is increased according to ${}^{(q+1)}d = {}^{(q)}d + 1$, and Subproblem Reinitialization updates the verification index set ${}^{(q+1)}\mathcal{G}^{(0)}$ according to the information provided by Global Compensation Prioritization. This procedure compensates for the insufficiency of the current verification sparsity level and its corresponding verification index set, and is referred to as the Compensation Process.
	
According to the definitions of GMR and UVR together with the corresponding \textbf{Observations}, the insufficiency of the current verification sparsity level mainly arises from two situations. First, multiple true DOAs may be contained within a single GMR due to the merging of adjacent mainlobes, whereas the previous Global Initialization or Subproblem Reinitialization fails to establish verification grid points for all of them. Second, some true DOAs may lie within UVRs, such that their corresponding verification grid points have not yet been initialized by the previous Global Initialization or Subproblem Reinitialization. It should be noted that these two situations may occur independently or simultaneously during a single execution of the SAPD Search Algorithm.
	
Based on the above analysis, since multiple GMRs and UVRs may coexist in the spatial spectrum, a compensation priority rule must be established for different GMR and UVR regions. Accordingly, Global Compensation Prioritization determines the corresponding compensation priorities according to the GMRs and UVRs identified from the spatial spectrum. Subproblem Reinitialization is then performed according to these priorities to generate the updated verification index set ${}^{(q+1)}\mathcal{G}^{(0)}$ for the subproblem corresponding to the verification sparsity level ${}^{(q+1)}d$.

\textbf{1) Global Compensation Prioritization:} We first describe the compensation priority rule in detail. To establish a reasonable compensation priority rule, it is necessary to impose constraints on the number of verification grid points that are allowed to be initialized within each region under the proposed framework. It should be emphasized that such constraints do not imply that the corresponding regions can contain only a limited number of true DOAs in a physical sense. Instead, they are algorithm-induced constraints constructed based on the structural characteristics of the spatial spectrum, and are introduced to characterize the upper bound on the number of verification grid points that may be initialized within each region under the proposed framework. Based on the above considerations, the following algorithm-induced constraints are introduced for GMRs and UVRs, respectively.
	
\textbf{GMR Capacity Constraint:} Each GMR $\mathcal{B}_i$ is allowed to initialize at most three verification grid points.
	
\textbf{Remarks:} Under the proposed algorithmic framework, each GMR is allowed to initialize at most three verification grid points, corresponding to the most complex mainlobe-merging scenario considered in this work. According to extensive observations of spatial spectrum structures, when a single GMR corresponds to more than three neighboring sources, it usually indicates that the angular separations among multiple sources simultaneously fall below the Rayleigh limit. In such cases, the corresponding GMR is typically formed by the severe merging of multiple adjacent mainlobes. Under this condition, the source discrimination capability of the proposed SAPD Search deteriorates significantly, making it difficult to guarantee reliable DOA estimation performance. Therefore, the upper bound on the number of verification grid points that can be initialized within a single GMR is set to three in the proposed framework.
	
\textbf{UVR Capacity Constraint:} Each UVR $\mathcal{V}_i$ is allowed to initialize at most one verification grid point.
	
\textbf{Remarks:} The UVR Capacity Constraint is established based on extensive observations of spatial spectrum structures. From both the energy distribution perspective and the definition of UVR, a UVR is characterized as a low-energy region located between two GMRs. If multiple sources simultaneously exist within the same UVR, their accumulated energy is more likely to generate a new local peak in that region, making it difficult to preserve the valley characteristic of the UVR. In such cases, the region is more likely to evolve into a new GMR rather than remain a UVR. In typical DOA estimation scenarios, where source power variations are relatively moderate and the spatial distribution of sources remains sufficiently sparse, multiple resolvable spectral peaks rarely appear within the same UVR. For extreme cases, such as multiple low-RCS (Radar Cross Section) targets simultaneously located between strong scatterers, resulting in highly imbalanced local energy distributions, the corresponding scenarios generally deviate from the conventional DOA estimation setting considered in this work and may introduce more challenging detection and resolution problems. Such situations are beyond the scope of the proposed framework. Based on the above observations of spatial spectrum structures, the maximum number of verification grid points that can be initialized within a single UVR is set to one under the proposed algorithmic framework.
	
In typical DOA estimation scenarios, when the mainlobes corresponding to multiple closely spaced incident signals merge together, the peak energy of the resulting GMR generally becomes more pronounced as the number of signals involved in the mainlobe merging increases. Consequently, GMRs formed by different degrees of mainlobe merging often exhibit noticeable differences in peak energy within the spatial spectrum. Based on this observed spectral characteristic, together with the previously introduced GMR Capacity Constraint, the proposed method further classifies different GMRs into several energy levels according to their peak energies. It should be emphasized that these energy levels do not directly correspond to the actual number of incident signals contained within a GMR. Instead, they are only used to determine whether different GMRs possess similar peak-energy levels and, consequently, whether they may exhibit a similar degree of mainlobe merging. This can serve as the basis for the subsequent compensation priority design.
	
We employ two thresholds, $\kappa_1$ and $\kappa_2$, to classify the peak energies $\tilde{P}_i$ of all GMRs in the normalized spatial spectrum into three levels. When $\tilde{P}_i < \kappa_1$, the corresponding GMR is classified as a Level-1 region. When $\kappa_1 \leq \tilde{P}_i < \kappa_2$, the corresponding GMR is classified as a Level-2 region. When $\kappa_2 \leq \tilde{P}_i \leq 1$, the corresponding GMR is classified as a Level-3 region. Here, Level-1, Level-2, and Level-3 denote three different energy levels. Accordingly, during Global Compensation Prioritization, Level-3, Level-2, and Level-1 regions are assigned high, medium, and low priority compensation levels, respectively. 
	
Because the spatial spectrum can be expressed as the linear superposition of the energies of multiple incident signals, the theoretical derivation is conducted under the equal-power signal assumption by considering the most conservative case in which all incident signals are completely overlapped. Under this assumption, the theoretical values are obtained as $\kappa_1 = 1/3$ and $\kappa_2 = 2/3$. In practical DOA estimation, however, target powers are generally unequal and multiple incident signals are rarely completely overlapped. Therefore, slightly larger empirical values, namely $\kappa_1 = 0.4$ and $\kappa_2 = 0.7$, are adopted to improve the robustness of the proposed algorithm. In our experiments, the proposed algorithm is not sensitive to small variations around these empirical values.

It should be noted that, since the energy-level classification is established based on the normalized peak spectrum energy $\tilde{P}$, at least one GMR always satisfy $\tilde{P}_i \geq \kappa_2$ and therefore be classified as a Level-3 region. The GMRs corresponding to the remaining energy levels, however, do not necessarily appear simultaneously within the same spatial spectrum. Consequently, the proposed energy-level classification is relative rather than absolute. Its purpose is to characterize the relative peak-energy relationships among different GMRs in the current spatial spectrum, as well as the potential differences in their corresponding mainlobe fusion degrees.
	
Based on the aforementioned GMR energy levels and the algorithmic constraints imposed on both GMRs and UVRs, the proposed compensation strategy establishes the following compensation rules. Since different candidate compensation regions correspond to different compensation priorities, the compensation priority rules are divided into three parts, the Hierarchical Structure Alignment Part, the UVR Compensation Part, and the Residual GMR Compensation Part.
	
According to the GMR Capacity Constraint, within the proposed algorithmic framework, the degree of mainlobe fusion associated with a GMR can be characterized by the number of initialized verification grid points contained within that GMR. Consequently, in the SAPD Search Scheme, GMRs belonging to the same energy level are generally considered to correspond to similar mainlobe fusion degrees. Therefore, the numbers of initialized verification grid points associated with such GMRs should be kept as consistent as possible.
	
\textbf{Hierarchical Structure Alignment Part:} The objective of the Hierarchical Structure Alignment Step is to ensure that the numbers of initialized verification grid points associated with different GMRs are consistent with the degrees of mainlobe fusion reflected by their corresponding energy levels. Based on this principle, GMRs belonging to the same energy level should contain, as much as possible, the same number of initialized verification grid points. If discrepancies exist among the numbers of initialized verification grid points associated with different GMRs within the same energy level, the largest number of initialized verification grid point currently observed within that energy level is taken as the reference, and compensation is performed for the remaining GMRs of the same level. For multiple GMRs requiring compensation within the same energy level, the compensation priority is determined according to their corresponding peak energies, where a higher peak energy corresponds to a higher compensation priority.
	
For GMRs belonging to different energy levels, new verification grid points are preferentially added to the GMRs with higher compensation priorities. Meanwhile, the differences in the numbers of initialized verification grid point among different energy levels are maintained to be consistent with their corresponding energy-level differences. Specifically, when the energy levels of two GMRs differ by one level, the numbers of initialized verification grid point associated with them are allowed to differ by one. When their energy levels differ by two levels, the corresponding initialized verification grid point numbers are allowed to differ by two.
	
It should be further noted that the compensation results obtained in this step do not necessarily correspond to the true numbers of physical sources. Instead, they are introduced solely to establish a reasonable fusion-structure initialization under the current spatial spectrum. For example, when all GMRs are classified into the same energy level, even if some GMRs may actually contain more physical sources, no additional compensation is required as long as the initialized verification grid point numbers are consistent among GMRs of the same level. In such a case, these GMRs are considered to exhibit the same degree of mainlobe fusion within the proposed framework. 
	
If all GMRs have satisfied the Hierarchical Structure Alignment rule through the Compensation Process, yet SAPD-Guided Search still fails to obtain a solution satisfying the prescribed recovery error tolerance, verification grid point compensation is subsequently performed within the UVRs.
	
\textbf{UVR Compensation Part:} All UVRs are ranked according to the energy values of their corresponding local minima, and their compensation priorities are determined accordingly. According to the UVR Capacity Constraint, a higher local-minimum energy generally indicates a higher probability that the corresponding UVR contains an uninitialized verification grid point. Consequently, the corresponding compensation priority is assigned to be higher.
	
If no UVR exists in the spatial spectrum, or if UVR compensation has been completed, yet SAPD-Guided Search still fails to obtain a solution satisfying the prescribed recovery error tolerance, verification grid point compensation is subsequently performed within the GMRs according to the Residual GMR Compensation rule.
	
\textbf{Residual GMR Compensation Part:} At this stage, all UVRs in the spatial spectrum have already been compensated and validated. Meanwhile, the Hierarchical Structure Alignment part is intended only to establish fusion-structure consistency among different GMRs, and therefore does not guarantee that all physical sources within each GMR have corresponding initialized verification grid points. Consequently, the remaining uninitialized verification grid points are still more likely to be associated with one or more GMRs. Based on this observation, the compensation priority is further determined according to the peak energies of the GMRs in descending order.
	
The above procedure constitutes the Global Compensation Prioritization process. It is executed only once before entering the Main Loop to establish the compensation priorities. During the Compensation Process in the Main Loop, Subproblem Reinitialization sequentially updates the verification index set according to the established compensation priorities until the prescribed recovery error tolerance is satisfied.
	
\textbf{Remarks:} It should be noted that the above compensation rule is derived under the equal-power signal assumption. When multiple incident signals $(K \ge 3)$ exhibit significant power differences, the corresponding spatial spectrum morphology may change considerably, and the associated compensation rule requires further analysis. This situation involves a more challenging multi-source unequal-power DOA estimation scenario, which is beyond the scope of this paper. Therefore, it is not further discussed herein.

\textbf{2) Subproblem Reinitialization:} According to the compensation priorities determined by Global Compensation Prioritization, Subproblem Reinitialization updates the verification index set $^{(q+1)}\mathcal{G}^{(0)}$ at each Compensation Process. For a GMR, if fewer than three verification grid points have been initialized, the newly introduced verification grid point is generated according to the same initialization rule adopted in Global Initialization. Once the third verification grid point is introduced into the GMR, its initialization position is directly assigned to the peak location of the corresponding GMR. For a UVR, the newly introduced verification grid point is initialized at the midpoint of the two UVR boundaries.

\subsection{Final DOA Estimation Output and The overall Framework of the SAPD Search Scheme}

After obtaining the on-grid solution ${}^{(q^*)}\mathcal{G}^{(j^*)}$, where $q^*$ denotes the final iteration index of the Main Loop, the corresponding off-grid parameter can be further estimated using any conventional off-grid refinement method, thereby obtaining the final DOA estimates. Since proposing a new method for estimating the off-grid parameter is not the primary objective of this work, a conventional local optimization procedure is adopted. Specifically, $\boldsymbol{x}_{ {}^{(q^*)}\mathcal{G}^{(j^*)} }$ and $\boldsymbol{\beta}_{ {}^{(q^*)}\mathcal{G}^{(j^*)} }$ are alternately updated using \eqref{x_theta_function} and \eqref{spatial_angular_pseudo_dev} until convergence. The final estimation output is 
\begin{equation}
	\hat{\boldsymbol{\vartheta}} = \boldsymbol{\vartheta}_{{}^{(q^*)}\mathcal{G}^{(j^*)}} + \boldsymbol{\beta}_{ {}^{(q^*)}\mathcal{G}^{(j^*)} }
\end{equation}

The overall procedure of the proposed SAPD Search Algorithm is summarized as follows. First, the spatial spectrum is constructed from the observation data using Bartlett beamforming. Global Initialization is then performed to identify the GMRs and UVRs in the spatial spectrum and generate the initial verification index set ${}^{(0)}\mathcal{G}^{(0)}$ together with its corresponding verification sparsity level ${}^{(0)}d$. Subsequently, Global Compensation Prioritization establishes the compensation priorities of the GMRs and UVRs according to the spatial spectrum structure.
	
The algorithm then enters the Main Loop, where SAPD-Guided Search solves the fixed-sparsity subproblem corresponding to the current verification sparsity level. If the obtained solution does not satisfy the prescribed recovery error tolerance, the verification sparsity level is increased according to ${}^{(q+1)}d = {}^{(q)}d + 1$, and Subproblem Reinitialization updates the verification index set ${}^{(q+1)}\mathcal{G}^{(0)}$ according to the established compensation priorities. This process continues until an on-grid solution ${}^{(q^*)}\mathcal{G}^{(j^*)}$ satisfying the prescribed recovery error tolerance $\epsilon$ is obtained for the first time.
	
Finally, the obtained on-grid solution is used to initialize the DOA Refinement stage, where the corresponding off-grid parameter is estimated to further improve the DOA estimation accuracy. The complete procedure of the proposed SAPD Search Algorithm is summarized in \ref{alg:SAPD_Search_Scheme}.
	
\begin{algorithm}
	\caption{SAPD Search Algorithm}
	\label{alg:SAPD_Search_Scheme}
		
	\begin{algorithmic}[1]
			
		\Require Observation data $\mathbf{y}\in\mathbb{C}^{M}$
			
		\Ensure Estimated DOAs $\hat{\boldsymbol{\vartheta}}$
			
		\State Perform Global Initialization via Algorithm~\ref{alg:1} to get the initial verification index set ${}^{(0)}\mathcal{G}^{(0)}$ and verification sparsity level ${}^{(0)}d$
			
		\State Perform Global Compensation Prioritization to obtain the compensation prioritization for the GMRs and UVRs
			
		\For{$q=0,1,\cdots$}
			
		\State Execute SAPD-Guided Search via Algorithm~\ref{SAPD_Search_Step}
			
		\State Obtain the on-grid solution ${}^{(q)}\mathcal{G}^{(j^*)}$ and flag $F_s$
			
		\If{$F_s=1$}
			
		\State Update the verification sparsity level ${}^{(q+1)}d={}^{(q)}d+1$

		\State Perform Subproblem Reinitialization by the compensation prioritization to obtain
		${}^{(q+1)}\mathcal{G}^{(0)}$
			
		\Else
		\State Set $q^*=q$
			
		\State Obtain the on-grid solution ${}^{(q^*)}\mathcal{G}^{(j^*)} $
			
		\State Perform DOA Refinement to estimate off-grid parameter $\boldsymbol{\beta}_{ {}^{(q^*)}\mathcal{G}^{(j^*)} }$.
			
		\State $\hat{\boldsymbol{\vartheta}} = \boldsymbol{\vartheta}_{{}^{(q^*)}\mathcal{G}^{(j^*)}} + \boldsymbol{\beta}_{ {}^{(q^*)}\mathcal{G}^{(j^*)} }$
			
		\State \textbf{break}
			
		\EndIf
			
		\EndFor
			
		\State \Return $\hat{\boldsymbol{\vartheta}}$
			
	\end{algorithmic}
		
\end{algorithm}

\subsection{Computational complexity analysis}
	
This subsection analyzes the computational complexity of the proposed SAPD Search Algorithm. As shown in \textbf{Alg.~\ref{alg:SAPD_Search_Scheme}}, the computational complexity of the proposed algorithm is mainly dominated by SAPD-Guided Search and the two least-squares computations in the DOA Refinement stage, i.e., \eqref{x_theta_function} and \eqref{spatial_angular_pseudo_dev}. Since the DOA Refinement stage is not the primary focus of this work, only the computational complexity of a single search iteration in SAPD-Guided Search and that of the overall Main Loop are analyzed. 
	
The computational complexity of each search iteration in SAPD-Guided Search is $O(M{}^{(q)}d^2)$. In contrast, the per-iteration computational complexity of conventional $\ell_1$-norm minimization methods and sparse Bayesian methods is typically $O(G^3)$. In sparse DOA estimation, it generally holds that $G \gg {}^{(q)}d$, $G > M$, and $M > {}^{(q)}d$. Meanwhile, the per-iteration computational complexity of IAA-APES is $O(M^2G)$. Therefore, the computational complexity of a single search iteration in SAPD-Guided Search is significantly lower than those of the aforementioned algorithms. Accordingly, the overall computational complexity of the Main Loop can be expressed as
\begin{equation}
	O(\sum_{q = 0}^{q^*}{}^{(q)}\tilde{j}M{}^{(q)}d^2),
\end{equation}
where ${}^{(q)}\tilde{j}$ denotes the maximum number of search iterations performed by SAPD-Guided Search under the verification sparsity level ${}^{(q)}d$. Since ${}^{(q)}d \ll G$ generally holds, the overall computational complexity of the Main Loop remains significantly lower than that of existing overcomplete dictionary-based optimization methods.
	
Since the overall computational complexity of different algorithms is affected by multiple factors, a fair comparison cannot be made solely based on theoretical complexity. Therefore, the practical computational efficiency of the proposed algorithm will be further compared with existing methods in terms of runtime in the experimental section.

\section{Numerical Simulations and Experimental Validation} \label{Sec6}

This section provides a comprehensive validation of the proposed algorithm through both numerical simulations and real-world experiments. First, visual illustrations are presented to demonstrate the key procedures of the SAPD Search Algorithm. Subsequently, the estimation performance of the proposed algorithm is systematically evaluated in terms of estimation accuracy, angular resolution, computational complexity, and robustness under challenging scenarios. Furthermore, the effectiveness of the proposed Compensation Rules is verified through a representative example, followed by a sensitivity analysis of the recovery error tolerance $\epsilon$. Finally, the effectiveness of the proposed algorithm in practical applications is validated using real radar data collected by a TI AWR1843 mmWave radar system.

All experiments are conducted on a PC equipped with a 2.6 GHz Intel Core i7 processor and 16 GB RAM. The proposed and comparable algorithms are implemented in Python. The source amplitudes are generated from a normal distribution, i.e., $s(t) \sim \mathcal{N}(30, 5)$, unless stated otherwise. The spatial grid is defined over $\boldsymbol{\theta} = [-90^\circ, 90^\circ]$ with a grid interval of  $\Delta \theta = 1^\circ$. In all experiments, no prior knowledge of the number of incident sources is provided to the proposed method.

To quantify the estimation accuracy, the root-mean-square error (RMSE) is adopted as the primary performance metric. We define the RMSE as follows
\begin{equation}
	\mathrm{RMSE} = \sqrt{ \frac{1}{M_t K} \sum_{i = 1}^{M_t} \| \hat{\boldsymbol{\theta}}^{(i)} - \boldsymbol{\theta}^* \|_2^2 },
\end{equation}
where $\hat{\boldsymbol{\theta}}^{(i)}$ represents the estimate obtained in the $i$-th Monte Carlo trial, and $M_t$ is the total number of trials.

\subsection{ Visual Illustrations of the Key Procedures}

This subsection provides visual illustrations of the SAPD-Guided Search and the Compensation Procedure in the proposed SAPD Search Algorithm. An 8-element ULA with half-wavelength inter-element spacing is adopted throughout this subsection.

\begin{figure}[!t]
	\centering
	\subfloat[\small]
	{\includegraphics[width=3.0 in]{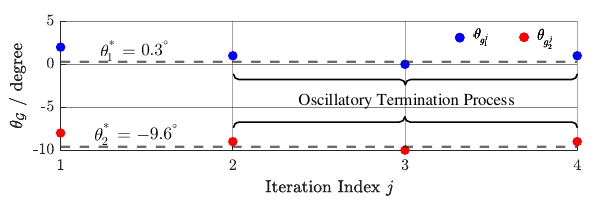}
		\label{A1_Iter_theta}} 
	\\
	\subfloat[\small]{\includegraphics[width=3.0 in]{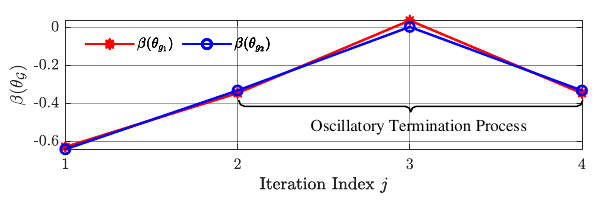}
		\label{A1_Iter_beta}} 
	\\
	\subfloat[\small]{\includegraphics[width=3.0 in]{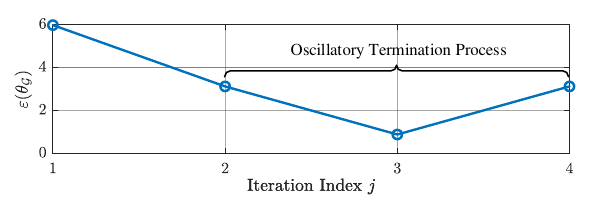}
		\label{A1_Iter_ep}} 
	\caption{Visualization of the proposed SAPD-Guided Search. (a) Evolution of candidate angles. (b) Evolution of grid bias. (c) Evolution of recovery error.}
	\label{Behavior_SGS}
\end{figure}

\textbf{Example A1 (Search Process of the SAPD-Guided Search)}: This example illustrates the search process of the proposed SAPD-Guided Search. A representative scenario containing two true DOAs, located at $\theta_1^* = 0.3^\circ$ and $\theta_2^* = -9.6^\circ$, is considered. The SNR is set to $15$dB.

Fig. \ref{Behavior_SGS}(a) illustrates the evolution of the candidate angles $\boldsymbol{\vartheta}_{\mathcal{G}^{(j)}}$ during the search process. It can be observed that the SAPD-Guided Search reaches a solution satisfying the SAPD Constraints after only four search steps. The subsequent three search steps enter an oscillatory stage, and the search is finally terminated according to the proposed oscillation-based termination criterion. 

Fig. \ref{Behavior_SGS}(b) shows the corresponding evolution of the grid bias $\boldsymbol{\beta}({\boldsymbol{\vartheta}_{\mathcal{G}^{(j)}}})$. It can be observed that, from the second to the third search step, the sign of each component changes while its magnitude further decreases. This behavior indicates that the search result gradually satisfies the SAPD Constraints.

Fig. \ref{Behavior_SGS}(c) presents the corresponding recovery error $\varepsilon({\boldsymbol{\vartheta}_{\mathcal{G}^{(j)}}})$. The recovery error continuously decreases during the first three search steps and then exhibits periodic oscillation in the oscillatory termination stage, which is consistent with the search trajectory shown in Fig. \ref{Behavior_SGS}(a).

\textbf{Example A2 (Illustration of the Compensation Process)}: This example provides a visual illustration of the proposed Compensation Process. A representative scenario requiring compensation is considered, where two true DOAs are located at $\theta_1^* = 0.3^\circ$ and $\theta_2^* = -7.6^\circ$. The SNR is set to $15$dB.

Fig. \ref{Comp_B}(a) illustrates the verification index sets before and after the Compensation Process. The verification index set obtained by the Global Initialization is ${}^{(0)}\mathcal{G}^{(0)} = \{ 89 \}$, corresponding to the verification grid point  $\theta_{{}^{(0)}g^{(0)}_1} = -2^\circ$. After the Compensation Process, the verification sparsity level is increased from ${}^{(0)}d$ to ${}^{(1)}d$, and Subproblem Reinitialization updates the verification index set to ${}^{(1)}\mathcal{G}^{(0)} = \{ 85, 93 \}$, corresponding to the verification grid points $\theta_{{}^{(1)}g^{(0)}_1} = -6^\circ$ and $\theta_{{}^{(1)}g^{(0)}_2} = 2^\circ$.

\begin{figure}[!t]
	\centering
	\subfloat[\small]
	{\includegraphics[width=3.0 in]{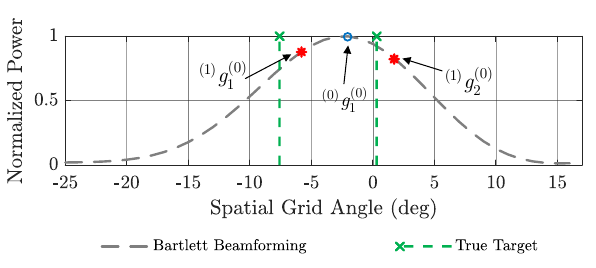}
		\label{A2_SS}} 
	\\
	\subfloat[\small]{\includegraphics[width=3.0 in]{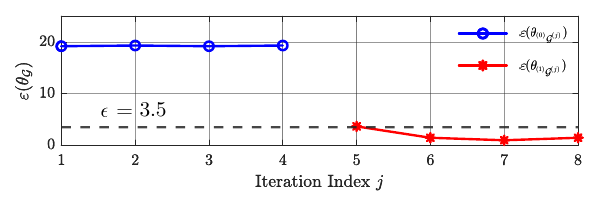}
		\label{A2_Iter_C}} 
	\\
	\caption{Visualization of the proposed Compensation procedure. (a) Subproblem initialization. (b) Recovery error during the Compensation procedure.}
	\label{Comp_B}
\end{figure}

Fig. \ref{Comp_B}(b) shows the corresponding recovery error $\varepsilon{\boldsymbol{\theta}_{{}^{(k)}\mathcal{G}^{(j)}}}$ during the search process. In Fig. \ref{Comp_B}(b), the first four search steps correspond to the SAPD-Guided Search initialized from the verification index set ${}^{(0)}\mathcal{G}^{(0)}$ obtained by the Global Initialization. It can be observed that, although the search converges to a local minimum satisfying the SAPD Constraints, the corresponding recovery error fails to satisfy the prescribed recovery error tolerance $\epsilon$. Consequently, the Main Loop increases the verification sparsity level and performs Subproblem Reinitialization. The fifth to eighth search steps correspond to the SAPD-Guided Search after Subproblem Reinitialization. It can be observed that the obtained solution not only satisfies the SAPD Constraints but also satisfies the recovery error tolerance $\epsilon$ for the first time. Therefore, this solution is accepted as the final on-grid DOA estimation result.

\subsection{Estimation Performance}

The overall performance of the proposed SAPD Search Algorithm is evaluated in this subsection. Unless otherwise specified, all RMSE results are computed over $M_t = 1000$ 
independent Monte Carlo trials. 

The overall performance of the proposed SAPD Search Algorithm is evaluated from six aspects, including estimation accuracy under different SNRs, angular resolution, array size, source number, multi-source scenarios, and source power imbalance.

Unless otherwise specified, the parameters of the proposed SAPD Search Algorithm are set as follows. The recovery tolerance is set to $\epsilon = 7.5$, which is determined according to the transmitted signal $s(t)$. The compensation term for beamwidth is set to $\delta = 2^\circ$. The two GMR energy-level thresholds are empirically set to $\kappa_1 = 0.4$ and $\kappa_1 = 0.7$. 

Since the primary objective of this work is to reduce the computational complexity of sparse DOA estimation while maintaining high estimation accuracy, rather than to develop a new off-grid parameter estimation method, the comparison algorithms are selected according to this principle. For on-grid sparse DOA estimation, the trimmed LASSO solved by GSM \cite{amir2021trimmed} is adopted as the high-accuracy benchmark, while IAA-APES \cite{IAA_APES} and OMP \cite{yang2018sparse} are selected as representative low-complexity sparse DOA estimation algorithms. To eliminate the influence of off-grid errors on the comparison, the same off-grid parameter refinement adopted in this work is applied to all on-grid algorithms. For off-grid sparse Bayesian methods, OGSBI \cite{OGSBI} and GE \cite{GE} are included for comparison. For gridless methods, the original ANM \cite{yang2018sparse} and the low-complexity yet high-accuracy VALSE \cite{VALSE} algorithm are adopted as representative comparison methods.

In addition, the deterministic Cram\'{e}r–Rao Bound (CRLB) \cite{CRB} is adopted as the theoretical performance benchmark for evaluating the estimation accuracy of all algorithms. Since the deterministic CRB is a well-established result, its formulation is not repeated here and can be found in \cite{CRB}.

\begin{figure*}[!t]
	\centering
	\subfloat[\small]
	{\includegraphics[width=3.1 in]{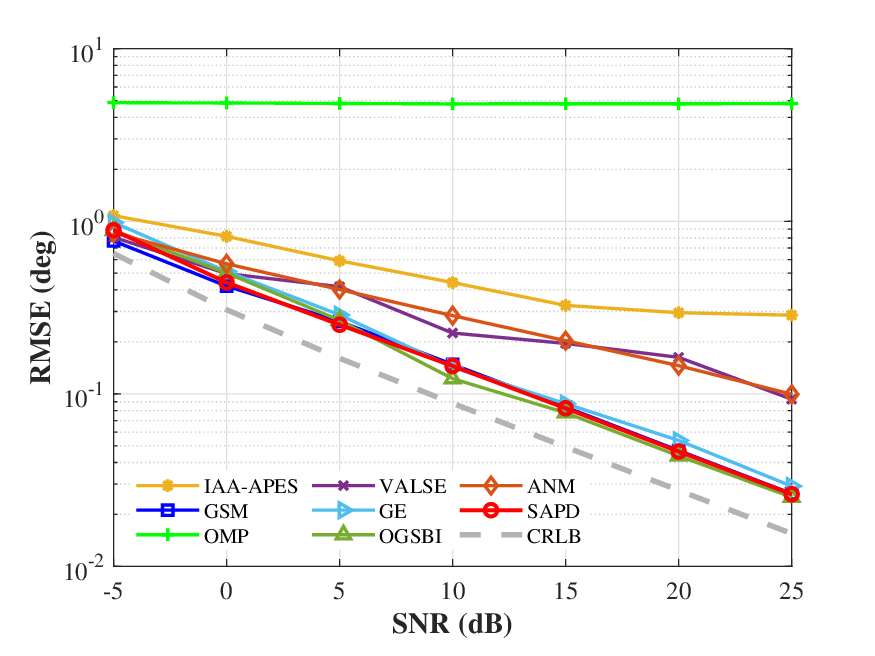}
		\label{B1_RMSE_SNR}} 
	\hfil
	\subfloat[\small]{\includegraphics[width=3.1 in]{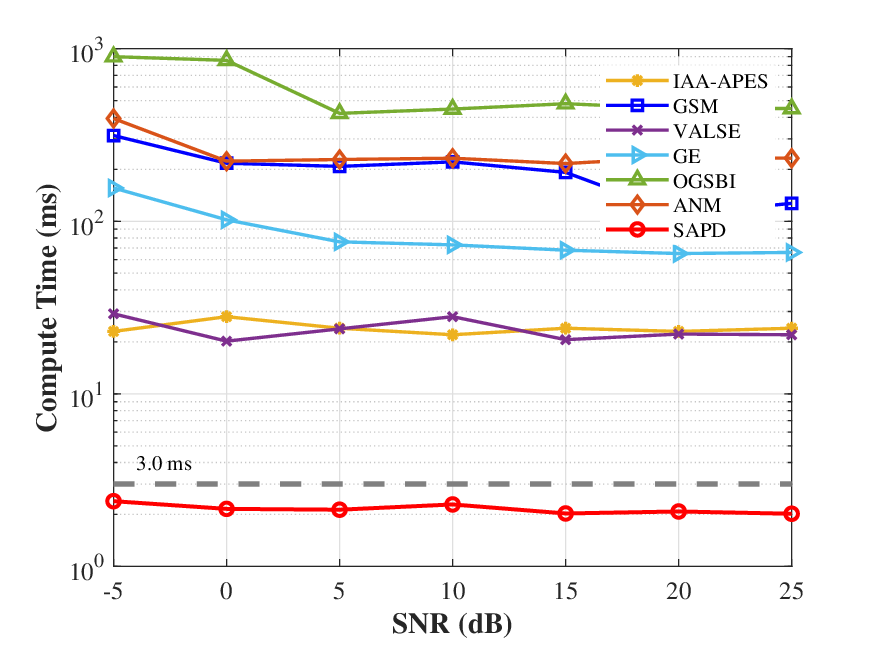}
		\label{B1_RMSE_SNR_CT}} 
	\caption{RMSE performance comparison versus SNR for two incident sources $\theta_1^* = 0.12^\circ$, $\theta_2^*  = 7.78^\circ$. (a) RMSE versus SNR (b) Compute Time.}
	\label{B1}
\end{figure*}

\textbf{Example B1 (RMSE Performance versus SNR)}: This example evaluates the RMSE performance of the proposed SAPD Search Algorithm under varying SNR conditions. The two fixed incident angles are set to $0.12^\circ$ and $7.78^\circ$, respectively. The SNR varies from $-5$dB to $25$dB with a step size of $5$dB. The number of array sensors of the ULA is fixed at $8$.

As shown in Fig. \ref{B1}(a), the RMSE of the proposed SAPD Search Algorithm continuously decreases as the SNR increases, indicating that its DOA estimation accuracy improves with increasing signal-to-noise ratio. Among the compared algorithms, the proposed SAPD Search Algorithm achieves RMSE performance comparable to that of GSM, OGSBI, and GE, while consistently outperforming IAA, VALSE, and ANM. According to the RMSE results of OMP, it fails to provide valid DOA estimation under the considered experimental scenario.

Fig. \ref{B1}(b) compares the average runtime of different algorithms under the same experimental conditions. Since OMP fails to produce valid DOA estimates, its runtime is not reported. It can be observed that the proposed SAPD Search Algorithm requires less than $3$ms on average for a single estimation. The runtimes of IAA and VALSE are approximately $20–30$ms, whereas OGSBI exhibits the highest computational cost. Although GE reduces the computational complexity to some extent through the grid evolution strategy, its runtime remains significantly higher than that of the proposed SAPD Search Algorithm, exceeding it by at least one order of magnitude. Moreover, the proposed SAPD Search Algorithm exhibits relatively low sensitivity to SNR variations.

By jointly considering Fig. \ref{B1}(a) and Fig. \ref{B1}(b), it can be concluded that the proposed SAPD Search Algorithm achieves estimation accuracy comparable to that of high-accuracy algorithms while significantly reducing the computational runtime.

\begin{figure*}[!t]
	\centering
	\subfloat[\small]
	{\includegraphics[width=3.1 in]{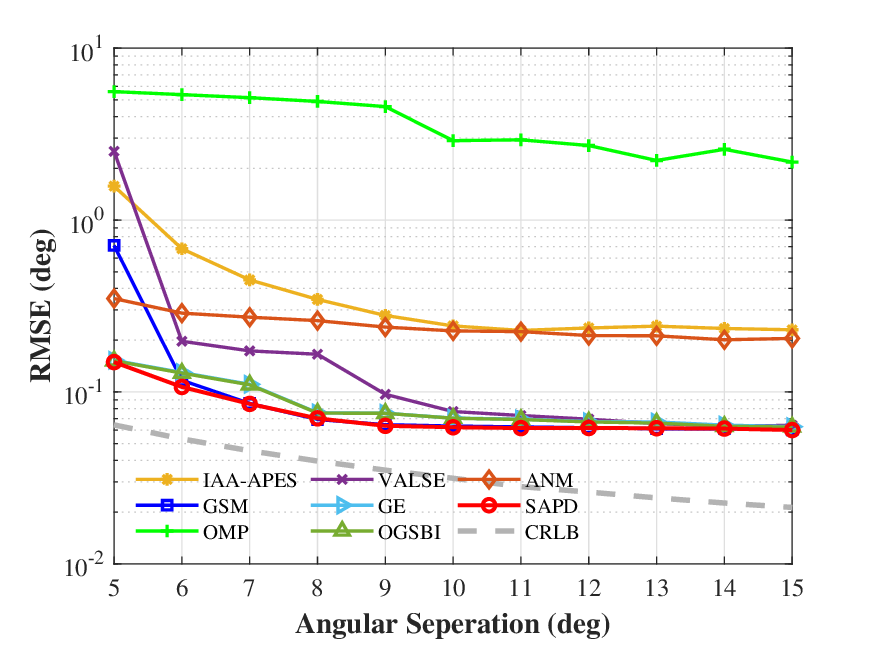}
		\label{B2_RMSE_Sep}} 
	\hfil
	\subfloat[\small]{\includegraphics[width=3.1 in]{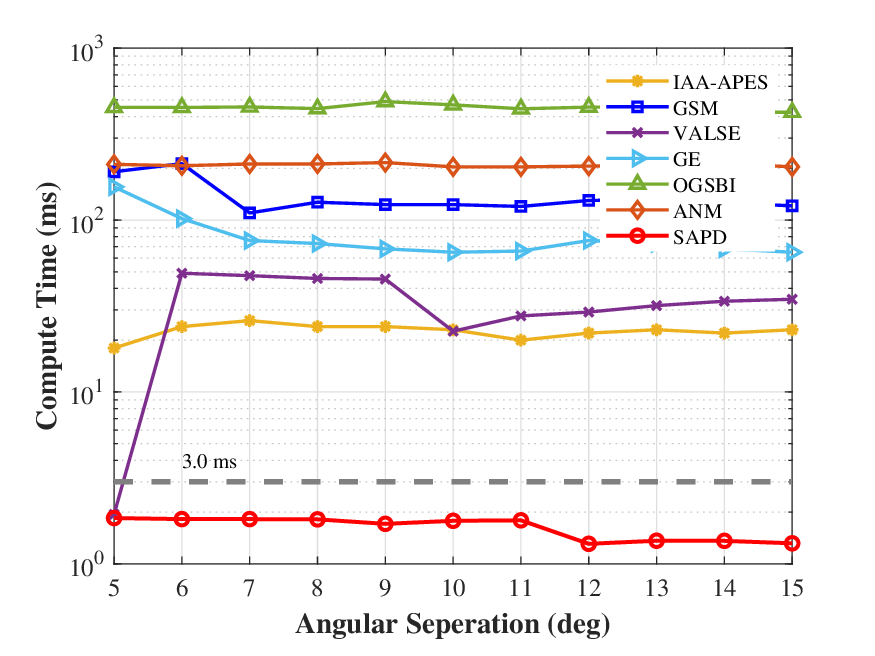}
		\label{B2_CT_Sep}} 
	\caption{Estimation performance versus angular separation. (a) RMSE (b) Compute Time.}
	\label{B2}
\end{figure*}

\textbf{Example B2 (Angular Resolution Performance)}: This example evaluates the DOA estimation performance of the proposed SAPD Search Algorithm under different angular separations, thereby validating its angular resolution capability. The SNR is fixed at 15 dB. The first true DOA is set as $\theta_1^* = 0^\circ + \zeta$, where $\zeta$ is a randomly generated off-grid parameter whose range is $[-\Delta\theta / 2, \, \Delta\theta / 2]$. The second true DOA is set as $\theta_2^* = \theta_1^* + \Delta \gamma$, where $\Delta \gamma$denotes the angular separation between the two signals. The angular separation varies from $5^\circ$ to $15^\circ$ with a step size of $1^\circ$. The number of array sensors of the ULA is fixed at $8$.

As shown in Fig. \ref{B2}(a), the RMSE of the proposed SAPD Search Algorithm continuously decreases as the angular separation increases. Throughout the entire angular separation range, its RMSE remains the closest to the CRLB. Among the low-complexity algorithms, the RMSE of both IAA and VALSE is consistently higher than that of the proposed SAPD Search Algorithm, while OMP still fails to provide valid DOA estimation under the considered experimental scenario. For high-accuracy but computationally expensive algorithms, including GSM, OGSBI, GE, and ANM, their RMSE is lower than that of the proposed SAPD Search Algorithm within the Rayleigh limit. However, beyond the Rayleigh limit, the proposed SAPD Search Algorithm achieves the same RMSE performance as GSM, OGSBI, GE, and VALSE.

\begin{figure*}[!t]
	\centering
	\subfloat[\small]
	{\includegraphics[width=3.1 in]{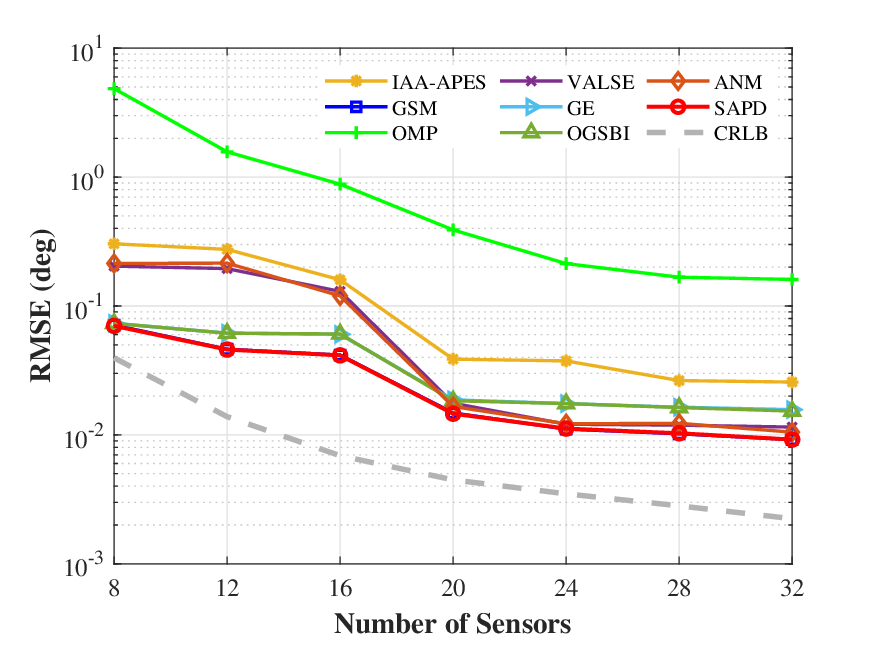}
		\label{B3_RMSE_Ap}} 
	\hfil
	\subfloat[\small]{\includegraphics[width=3.1 in]{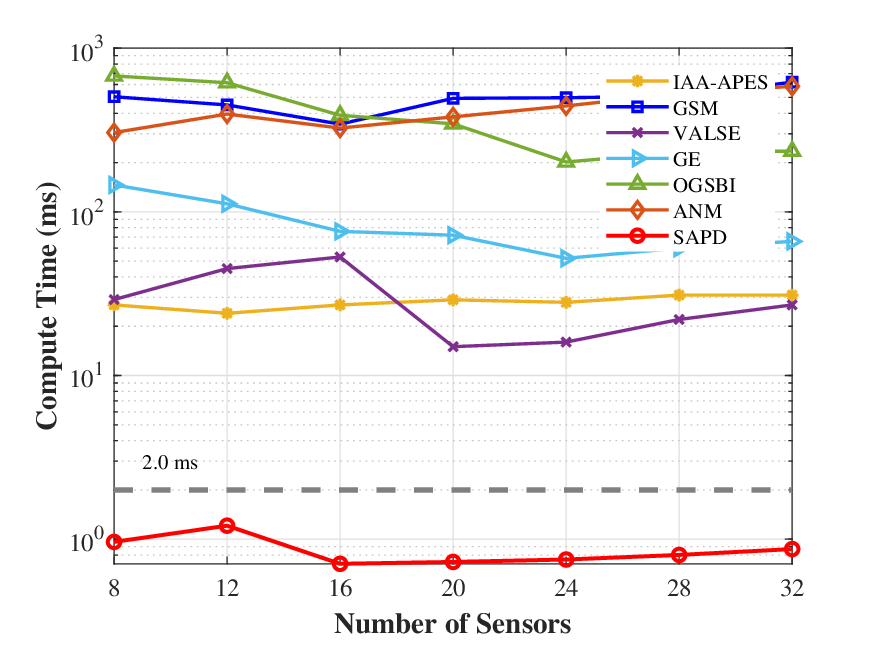}
		\label{B3_CT_Ap}} 
	\caption{Estimation performance versus the number of sensors. (a) RMSE (b) computational time.}
	\label{B3}
\end{figure*}

As shown in Fig. \ref{B2}(b), the average runtime of different algorithms under varying angular separations is presented. Since OMP fails to provide valid DOA estimation, its runtime is not reported. Throughout the entire angular separation range, the average runtime of the proposed SAPD Search Algorithm remains below 3 ms, which is at least one order of magnitude lower than that of the other compared algorithms. For VALSE, the estimation fails when the angular separation is $5^\circ$. Therefore, the corresponding runtime at this point is not considered meaningful for comparison. In addition, when the angular separation exceeds $11^\circ$ , the runtime of the proposed SAPD Search Algorithm exhibits a slight decrease. This is because the Compensation Process is no longer required under this condition, resulting in a further reduction in the overall computational cost.

By jointly considering Fig. \ref{B2}(a) and Fig. \ref{B2}(b), it can be observed that the proposed SAPD Search Algorithm maintains high angular resolution capability while preserving low computational complexity, thereby achieving high-accuracy DOA estimation.

\textbf{Example B3 (RMSE Performance versus the number of senosrs)}: This example evaluates the DOA estimation performance of the proposed SAPD Search Algorithm under different numbers of array sensors. The SNR is fixed at $15$dB. The two true DOAs are fixed at $\theta_1^* = 0.12^\circ$ and $\theta_2^*  = 7.88^\circ$. The number of array sensors $M$ is increased from $8$ to $32$ with a step size of $4$.

As shown in Fig. \ref{B3}(a), the RMSE of the proposed SAPD Search Algorithm continuously decreases as the number of array sensors increases. Throughout the entire range of array sizes, its RMSE remains the closest to the CRLB. Among the low-complexity algorithms, the RMSE of both IAA and VALSE is consistently higher than that of the proposed SAPD Search Algorithm. Although OMP is able to provide valid DOA estimation when the number of array sensors becomes sufficiently large, its RMSE remains approximately one order of magnitude higher than that of the other compared algorithms.

As shown in Fig. \ref{B3}(b), the average runtime of different algorithms under varying numbers of array sensors is presented. Consistent with the previous two experiments, the proposed SAPD Search Algorithm is still able to achieve high-accuracy DOA estimation within $2$ms. It can also be observed that, when the number of array sensors is $8$ or $12$, the runtime is slightly higher due to the execution of the Compensation Process. As the number of array sensors increases from $16$ to $32$, the runtime exhibits an approximately linear growth. This trend is consistent with the computational complexity analysis presented in the previous section.

By jointly considering Fig. \ref{B3}(a) and Fig. \ref{B3}(b), it can be observed that, under varying numbers of array sensors, the proposed SAPD Search Algorithm is able to maintain low computational complexity while achieving high-accuracy DOA estimation.

\begin{figure*}[!t]
	\centering
	\subfloat[\small]
	{\includegraphics[width=3.1 in]{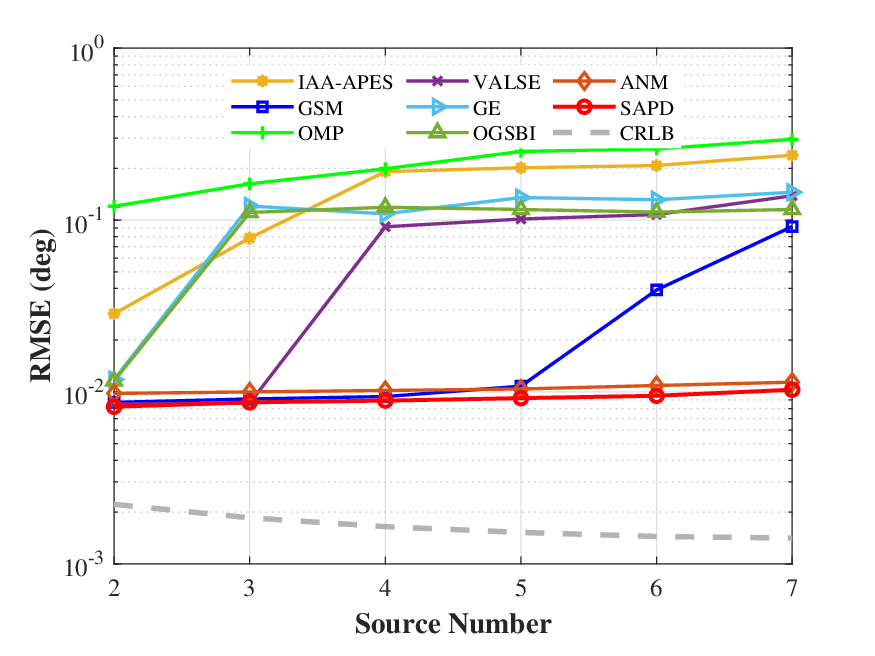}
		\label{B4_RMSE_Sn}} 
	\hfil
	\subfloat[\small]{\includegraphics[width=3.1 in]{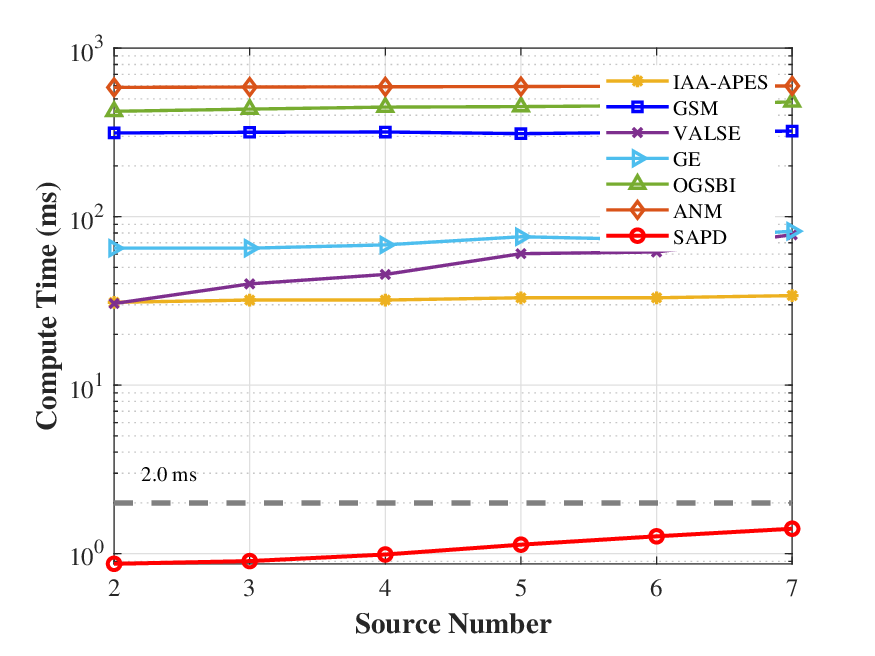}
		\label{B4_CT_Sn}} 
	\caption{Estimation performance versus Source Number. (a) RMSE (b) Compute Time.}
	\label{B4}
\end{figure*}

\textbf{Example B4 (RMSE Performance versus the Number of Sources)}: This example evaluates the DOA estimation performance of the proposed SAPD Search Algorithm under different numbers of incident sources. The SNR is fixed at $15$dB. To avoid the influence of an insufficient number of array sensors on the compared algorithms and to ensure a fair comparison, the number of array sensors of the ULA is fixed at $32$. The number of incident sources is increased from $2$ to $7$. All true DOAs are randomly generated within the range of $[-60^\circ, \, 60^\circ]$, while the angular separation between every two adjacent incident sources is fixed at $15^\circ$.

As shown in Fig. \ref{B4}(a), the RMSE of the proposed SAPD Search Algorithm increases slightly as the number of incident sources increases, following a trend similar to that of the other compared algorithms. Nevertheless, throughout the entire experimental range, the RMSE of the proposed SAPD Search Algorithm remains the closest to the CRLB. Among the low-complexity algorithms, the RMSE of OMP is consistently approximately one order of magnitude higher than that of the proposed SAPD Search Algorithm. Moreover, when $K \ge 3$, the RMSE of both IAA and VALSE increases significantly and remains approximately one order of magnitude higher than that of the proposed SAPD Search Algorithm. For the high-accuracy algorithms, the RMSE of GSM also begins to increase noticeably when the number of incident sources $K > 5$.

As shown in Fig. \ref{B4}(b), the average runtime of different algorithms under varying numbers of incident sources is presented. The proposed SAPD Search Algorithm is still able to achieve high-accuracy DOA estimation within $2$ms, while its runtime is at least one order of magnitude lower than that of the other compared algorithms. It can also be observed that the runtime of the proposed SAPD Search Algorithm gradually increases as the number of incident sources increases. This trend is consistent with the computational complexity analysis presented in the previous section.

By jointly considering Fig. \ref{B4}(a) and Fig. \ref{B4}(b), it can be observed that, under varying numbers of incident sources, the proposed SAPD Search Algorithm is able to maintain low computational complexity while achieving high-accuracy DOA estimation.

\begin{figure}[!t]
	\centering
	\includegraphics[width=3.1in]{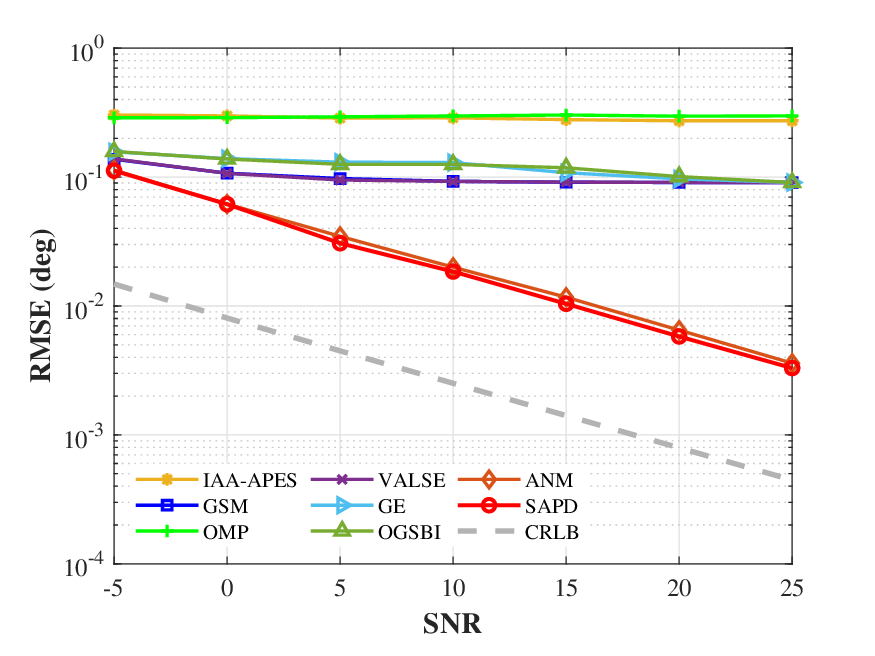}
	\caption{\centering{RMSE versus SNR for the seven-source scenario.}}
	\label{B6_RMSE_SNR_MS}
\end{figure}

\textbf{Example B5 (RMSE Performance versus SNR with Seven Sources)}: This example evaluates the DOA estimation performance of the proposed SAPD Search Algorithm under different SNRs in a multi-source scenario. To avoid the influence of an insufficient number of array sensors on the compared algorithms and to ensure a fair comparison, the number of array sensors of the ULA is fixed at $32$, following the same setting as in Example B4. The number of incident sources is fixed at $7$, and the signal generation procedure is identical to that used in Example B4. The SNR is varied from $-5$dB to $25$dB with a step size of $5$dB.

As shown in Fig. \ref{B6_RMSE_SNR_MS}, the RMSE of the proposed SAPD Search Algorithm continuously decreases as the SNR increases and remains the closest to the CRLB throughout the entire SNR range. Under this experimental scenario, ANM is the only compared algorithm that achieves estimation accuracy comparable to that of the proposed SAPD Search Algorithm, while the RMSE of all the other compared algorithms is significantly higher. These results demonstrate that the proposed SAPD Search Algorithm is able to maintain high-accuracy DOA estimation and exhibits strong robustness against noise even in multi-source scenarios.

\textbf{Example B6 (RMSE Performance versus Source Power Difference)}: This example evaluates the DOA estimation performance of the proposed SAPD Search Algorithm under different source power differences. This example is designed to simulate practical scenarios where two targets exhibit significantly different radar cross sections (RCSs), resulting in large differences in the received echo power.

The SNR is fixed at $15$dB, and the ULA consists of $8$ array sensors. The two true DOAs are fixed at $\theta_1^* = 0.12^\circ$ and $\theta_2^* = 7.88^\circ$. Let $A_1$ denote the amplitude of the first target echo. The amplitude of the second target echo is given by $A_2 = A_1 \cdot 10^{-\Delta P / 20}$, where $\Delta P$ denotes the power difference between the two targets. The value of $\Delta P$ varies from $0$dB to $10$dB with a step size of $2$dB.

As shown in Fig. \ref{B5_RMSE_Sd}, the RMSE of the proposed SAPD Search Algorithm gradually increases as the source power difference increases, and the other compared algorithms exhibit the same trend. Nevertheless, throughout the entire range of source power differences, the proposed SAPD Search Algorithm consistently achieves the lowest RMSE. These results demonstrate that the proposed SAPD Search Algorithm is robust against source power differences.

\begin{figure}[!t]
	\centering
	\includegraphics[width=3.1in]{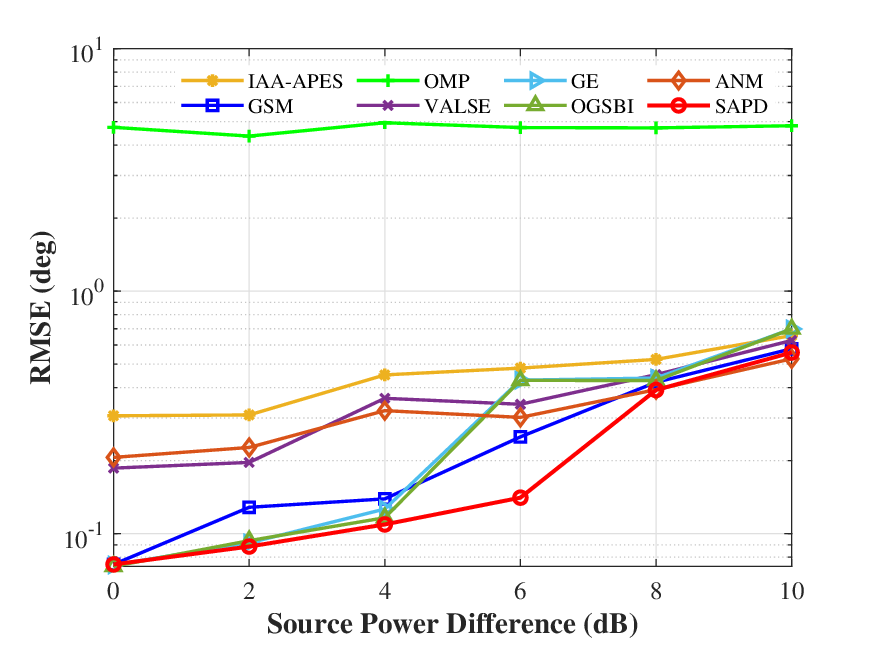}
	\caption{\centering{RMSE versus Source Power Differences.}}
	\label{B5_RMSE_Sd}
\end{figure}

From the above experimental results, it can be concluded that the proposed SAPD Search Algorithm consistently achieves high-accuracy DOA estimation while maintaining significantly lower computational complexity than existing sparse DOA estimation methods. Furthermore, the proposed algorithm exhibits stable estimation performance under different SNRs, array sizes, source numbers, and source power differences, demonstrating its robustness in practical radar applications.

\subsection{Effectiveness of the Compensation Rules}

This subsection presents a representative example to demonstrate that supplementing the verification index set according to the proposed Compensation Rules provides a more reliable verification index set than using Random Initialization for the subproblem reinitialization, thereby validating the necessity of the proposed Compensation Rules. In this experiment, an 8-element ULA is considered under an SNR of 15dB. The true DOAs are set as $\boldsymbol{\theta}^* = \{ -30^\circ, -20^\circ, -10^\circ, 37^\circ, 45^\circ \}$. 

As shown in Fig. \ref{Norm_ss_5_Target}, in the initial spatial spectrum, the three closely spaced sources $\{ -30^\circ, -20^\circ, -10^\}$ cause the source located at $-20^\circ$ to fail to form an independent peak and instead fall into the UVR generated by the adjacent sources. Consequently, it is ignored during the Global Initialization Step. Meanwhile, the two closely spaced sources $\{ 37^\circ, 45^\circ \}$ form only a single isolated-GMR, whose generalized beamwidth exceeds the criterion given in Observation 2. Therefore, the cardinality of the initial verification index set satisfies $\vert {}^{\{0\}} \mathcal{G} ^{\{0\}} \vert = 4$ rather than the true number of five incident sources.

\begin{figure}[!t]
	\centering
	\includegraphics[width=3.2in]{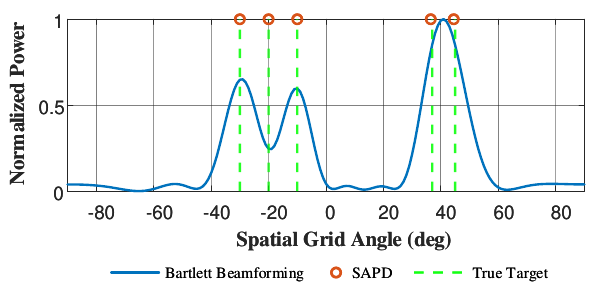}
	\caption{\centering{Normalized spatial spectra for five targets with DOAs of $\{ -30^\circ, -20^\circ, -10^\circ, 37^\circ, 45^\circ \}$ at SNR = 15dB}}
	\label{Norm_ss_5_Target}
\end{figure}

\begin{table}[H]
	\caption{Comparison of Compensation Strategies}
	\centering
	\begin{tabular}{ccccc}
		\toprule
		\textbf{Methods} &\textbf{Random Initialization} & \textbf{SAPD search} \\
		\midrule
		\textbf{RMSE} & $-$  & $\boldsymbol{0.2724^{\circ}}$ \\
		\bottomrule
	\end{tabular}
	\label{RMSE_5_Greedy}
\end{table}

After the first SAPD-Guided Search, the recovery error of the current solution still fails to satisfy the recovery tolerance $\epsilon$. Therefore, the Compensation Process is invoked in the Main Loop, where the missing verification grid points are supplemented according to the proposed Compensation Rules. Subsequently, the SAPD-Guided Search is performed again, yielding an estimation result that satisfies the recovery tolerance, as shown in Fig. \ref{Norm_ss_5_Target}.

Table \ref{RMSE_5_Greedy} summarizes the estimation results obtained using different initialization strategies. Since Random Initialization cannot provide an effective verification index set for the subproblem reinitialization, the SAPD-Guided Search fails to converge to the correct DOA solution, and thus its RMSE cannot be reported. In contrast, by adopting the proposed Compensation Rules, the algorithm successfully completes the DOA estimation for this scenario, achieving a final RMSE of $0.2724^{\circ}$.

These results demonstrate that, in complex multi-source scenarios, using Random Initialization to initialize the subproblem corresponding to a new verification sparsity level makes it difficult for the SAPD-Guided Search to obtain reliable DOA estimation results. In comparison, the proposed Compensation Rules can effectively supplement the missing verification grid points, thereby ensuring the correct execution of the subsequent SAPD-Guided Search.

\subsection{Sensitivity Analysis of the Recovery Error Tolerance $\epsilon$}

This subsection analyzes the sensitivity of the recovery error tolerance $\epsilon$. Since $\epsilon$ only affects the sparsity estimation process of the proposed SAPD Search Algorithm and does not directly participate in the subsequent DOA refinement, the source number estimation success rate is adopted as the evaluation metric. An estimation is regarded as successful if the estimated number of sources is equal to the true number of incident sources. Accordingly, the success rate is defined as
\begin{equation}
	SR = \frac{M_s}{M_t} \times 100\%
\end{equation}
where $M_s$ denotes the number of successful trials in the Monte Carlo simulations. The number of array sensors of the ULA is fixed at $8$.

\begin{figure}[!t]
	\centering
	\subfloat[\small]
	{\includegraphics[width=3.0 in]{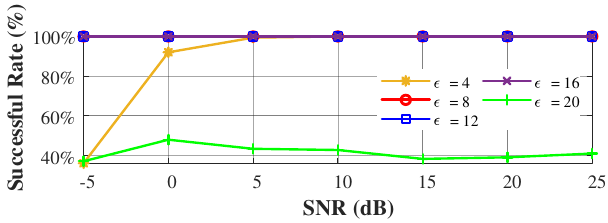}
		\label{D1_SNR_SR}} 
	\\
	\subfloat[\small]{\includegraphics[width=3.0 in]{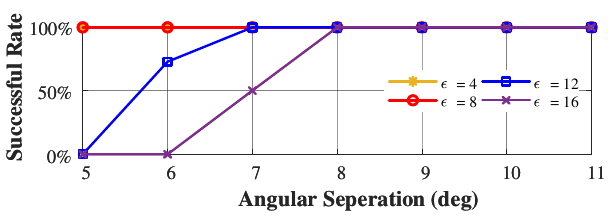}
		\label{D1_Angular_SR}} 
	\caption{Sensitivity analysis of the recovery error tolerance $\epsilon$. (a) Source number estimation success rate versus SNR. (b) Source number estimation success rate versus angular separation.}
	\label{D1}
\end{figure}

To analyze the influence of $\epsilon$ under different noise conditions, Fig. \ref{D1}(a) considers two fixed true DOAs located at $\theta_1^* = -0.12^\circ$ and $\theta_2^* = 7.88^\circ$. The SNR varies from $-5$dB to $25$dB with a step size of $5$dB. The source number estimation success rates corresponding to $\epsilon = 4, 8, 12, 16$ and $20$ are then evaluated.

As shown in Fig. \ref{D1}(a), when $\epsilon = 8, 12$, and $16$  the proposed algorithm is able to achieve stable source number estimation over the entire SNR range. In contrast, when $\epsilon = 20$, the recovery error tolerance becomes close to the recovery error corresponding to the verification sparsity level of one. As a result, the Main Loop tends to terminate prematurely, leading to a source number estimation success rate consistently below $50\%$. Conversely, when $\epsilon = 4$ the recovery error tolerance is overly restrictive. Under low-SNR conditions, the recovery error can hardly satisfy the recovery error constraint, resulting in a noticeable decrease in the source number estimation success rate.

To analyze the influence of $\epsilon$ under different angular resolution conditions, Fig. \ref{D1}(b) fixes the SNR at 15dB, while the angular separation between the two incident sources varies from $5^\circ$ to $11^\circ$ with a step size of $1^\circ$. Since $\epsilon = 20$ cannot provide stable source number estimation, the corresponding results are omitted from Fig. \ref{D1}(b).

As shown in Fig. \ref{D1}(b), when $\epsilon = 4$ and $\epsilon = 8$, the proposed algorithm is able to achieve stable source number estimation over the entire angular separation range. When $\epsilon = 12$, the source number estimation becomes unstable when the angular separation is smaller than or equal to $6^\circ$, and fails at an angular separation of $5^\circ$. For $\epsilon = 16$ unstable source number estimation begins to appear when the angular separation is smaller than or equal to $7^\circ$. It can be observed that, as the recovery error tolerance $\epsilon$ increases, the source number estimation success rate gradually decreases under small angular separation conditions.

By jointly considering Figs. \ref{D1}(a) and \ref{D1}(b), it can be observed that the selection of the recovery error tolerance $\epsilon$ should balance both the noise condition and the angular resolution. When $\epsilon$ is chosen around 8, the proposed algorithm achieves stable source number estimation performance under different experimental scenarios. Furthermore, together with the RMSE results presented in Section B, it can be concluded that the SAPD Search Algorithm is not sensitive to small variations in $\epsilon$. Therefore, this empirical parameter does not require precise tuning and can be calibrated in advance according to the signal model and the target application scenario.

\subsection{Experimental Validation with Real-World Data}

\begin{figure}[!t]
	\centering
	\includegraphics[width=2.6 in]{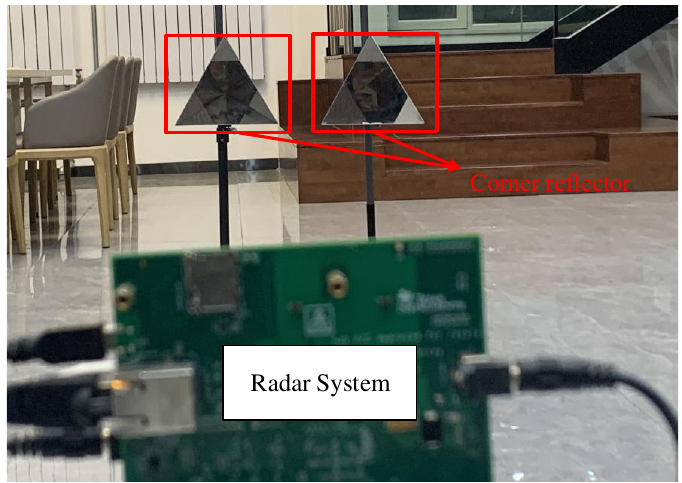}
	\caption{\centering{Experimental scenarios for hardware validation. Two corner reflectors with ground-truth DOAs at $\Delta \theta \in [-6.5^\circ, 0.5^\circ]$.}}
	\label{Hardware_SAPD_Multi_2_Scenario}
\end{figure}

\begin{figure}[!t]
	\centering
	\includegraphics[width=3.1in]{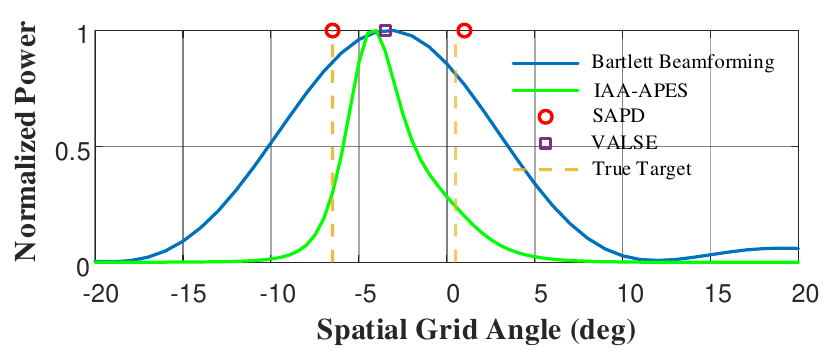}
	\caption{\centering{Normalized experimental spatial spectra for corner reflectors. Two sources at DOAs of $-6.5^\circ$ and $0.5^\circ$.}}
	\label{Hardware_SAPD_Multi_SS}
\end{figure}

In this subsection, we validate the performance of the proposed algorithm using raw radar data captured by a TI AWR1843 mmWave radar system. Corner reflectors are employed as targets and are positioned such that they fall within the same RD-cell to create a challenging super-resolution scenario, as illustrated in Fig. \ref{Hardware_SAPD_Multi_2_Scenario}. In this case, two corner reflectors are placed at $\{ -6.5^\circ, 0.5^\circ \}$. As shown in Fig. \ref{Hardware_SAPD_Multi_SS}, since the estimation results of GSM, OGSBI, GE, and ANM are close to those of the proposed SAPD algorithm, while OMP fails to effectively handle this scenario, only the estimation results of IAA-APES, VALSE, and the proposed SAPD algorithm are presented in the figure to maintain the clarity of the spatial spectra. It can be observed from Fig. \ref{Hardware_SAPD_Multi_SS} that, among the relatively low-complexity algorithms, i.e., IAA-APES, VALSE, and OMP, only the proposed SAPD algorithm can simultaneously identify the correct number of sources and achieve high-precision DOA estimation.

\section{Conclusion} \label{Sec7}

This paper addresses the problem that the use of large-scale matrix computations, such as the full overcomplete dictionary matrix and covariance-like matrices, in sparse DOA estimation leads to excessive computational complexity, making it difficult to satisfy the real-time requirements of automotive radar applications. An efficient Sparse DOA estimation method based on the SAPD Property is proposed. First, by analyzing the relationship between sparse DOA solutions and the spatial discrete grid, the SAPD Property exhibited by sparse DOA solutions on the discretized grid is revealed, and the corresponding SAPD Constraint is further constructed. Subsequently, the SAPD Constraint is incorporated into the conventional $\ell_0$-norm sparse recovery formulation, resulting in the SAPD-Constrained Sparse DOA Optimization Function. By exploiting the spatial structural information provided by the SAPD Property, the proposed optimization formulation transforms the large-scale exhaustive verification process required by conventional $\ell_0$-norm minimization due to unordered candidate solution combinations into an ordered grid search process with directional and step-size information. Based on the above optimization formulation and the SAPD Property, an SAPD Search Algorithm is further developed to efficiently solve the established optimization problem through a grid search strategy. By exploiting the SAPD Property exhibited by sparse DOA solutions, the proposed method avoids the involvement of large-scale matrix operations in the solving process, thereby maintaining low computational complexity during the search procedure. Numerical simulations and experimental results demonstrate that the proposed SAPD Search Algorithm simultaneously achieves millisecond-level computational efficiency, high-precision, and super-resolution DOA estimation. Therefore, the proposed method has considerable potential for practical applications in real-time millimeter-wave sensing systems and autonomous radar platforms.

\appendices
\section{Proof of Theorem \ref{D_Local_Minima}} \label{Proof:Th-1}

The optimization problem $(P_D)$ can be reformulated as the minimization of the penalized objective:
\begin{equation}
	\boldsymbol{x} = \mathop{\arg\min}\limits_{\boldsymbol{x}} \dfrac{1}{2} \| \boldsymbol{y} - \mathbf{A}(\boldsymbol{\vartheta}) \boldsymbol{x}(\boldsymbol{\vartheta}) \|_2^2 + \lambda \tau_k(\boldsymbol{x}(\boldsymbol{\vartheta}))
\end{equation}
where $\tau_k(\boldsymbol{x}(\boldsymbol{\vartheta})) = \sum_{i = k + 1}^N \vert \boldsymbol{x} \vert_{(i)} $ is the trimmed lasso of the ordered components $\vert \boldsymbol{x} \vert_{(1)} \geq \vert \boldsymbol{x} \vert_{(2)} \geq \cdots \vert \boldsymbol{x} \vert_{(G)}$.

For a fixed support set $\vert \mathcal{G} \vert = k$, the subproblem assumes that all components $\theta_i$ for $i \notin \mathcal{G}$ are zero, which implies $\tau(\boldsymbol{\vartheta}) = 0$. By choosing a penalty parameter $\lambda$ such that 
\begin{equation}
	\lambda \geq \| \boldsymbol{y} \|_2 \cdot \max_{i = 1, \dots, d} \| \boldsymbol{a} \|_2  
\end{equation}
the penalty term effectively enforces the cardinality constraint. Consequently, the local minimizers of the penalized objective satisfy the $k$-sparse requirement. $\hfill\blacksquare$

\makeatletter
\let\myorg@bibitem\bibitem
\def\bibitem#1#2\par{%
	\@ifundefined{bibitem@#1}{%
		\myorg@bibitem{#1}#2\par
	}{%
		\begingroup
		\color{\csname bibitem@#1\endcsname}%
		\myorg@bibitem{#1}#2\par
		\endgroup
	}%
}

\makeatother

\nocite{*}
\small
\bibliographystyle{IEEEtran}
\bibliography{Arxiv_TAES_SAPD.bib}

\end{document}